\documentclass[a4, 12pt,onecolumn]{IEEEtran}

\usepackage{tabularx}
\usepackage{multicol}
\usepackage{tabularx}
\usepackage{graphicx}
\usepackage{algorithmic}
\usepackage{algorithm}
\usepackage{array}

\usepackage{textcomp}
\usepackage{stfloats}
\usepackage{url}
\usepackage{verbatim}
\usepackage{graphicx}

\usepackage[square,numbers,sort&compress]{natbib}
\usepackage{graphicx}

\usepackage{float}
\usepackage{amsmath}
\usepackage{hyperref} 
\usepackage{amsthm}
\usepackage{amsmath}
\usepackage{mathrsfs}
\usepackage{amssymb} 
\usepackage{bbm,dsfont}   
\usepackage{enumerate} 

\usepackage[table]{xcolor}
\usepackage[font=small,labelfont=bf]{caption}
\usepackage{multirow}
\usepackage{tikz}
\usetikzlibrary{calc}
\usepackage{pifont}%

\usepackage{physics}

\usepackage{xcolor}

\usepackage{tabularx}

\newcounter{protocol}
\renewcommand{\theprotocol}{\arabic{protocol}}

\usepackage{subfig}

\hypersetup{
    colorlinks,
    linkcolor={blue!80!black},%
    citecolor={green!30!black},
    urlcolor={blue!80!black}
}

\newtheorem{theorem}{Theorem}
\theoremstyle{remark}	
\theoremstyle{remark}	
\theoremstyle{remark}	\newtheorem{proposition}[theorem]{Proposition}
\theoremstyle{remark} \newtheorem{definition}{Definition}

\theoremstyle{remark} \newtheorem{remark}{Remark}
\theoremstyle{remark} 

\theoremstyle{remark}	\newtheorem*{discussion*}{Discussion}
\title{On (Im)possibility of
Network Oblivious Transfer via Noisy Channels and Non-Signaling Correlations}
\author{Hadi Aghaee$^\dagger$, Christian Deppe$^\dagger$, Holger Boche$^*$\\
\textit{$^\dagger$Institute for Communications Technology, Technical University of Braunschweig,}\\
\textit{Braunschweig, Germany}\\
\textit{$^*$Chair of Theoretical Information Technology, Technical University of 
Munich,}\\ \textit{
Munich, Germany}\\
Email: (hadi.aghaee, christian.deppe)@tu-bs.de, boche@tum.de
}
\begin{document}
\date{}
\maketitle
\begin{abstract}
    This work investigates the fundamental limits of implementing network oblivious transfer via noisy multiple access channels and broadcast channels between honest-but-curious parties when the parties have access to general tripartite non-signaling correlations. By modeling the shared resource as an arbitrary tripartite non-signaling box, we obtain a unified perspective on both the channel behavior and the resulting correlations. Our main result demonstrates that perfect oblivious transfer is impossible. In the asymptotic regime, we further show that even negligible leakage cannot be achieved, as repeated use of the resource amplifies the receiver(s)'s ability to distinguish messages that were not intended for him/them. In contrast, the receiver(s)'s own privacy is not subject to a universal impossibility limitation.
\end{abstract}
\section{Introduction}
Oblivious transfer (OT) was first introduced by Rabin \cite{Rabin1}, who considered a communication scenario in which Alice sends a message to Bob, Bob successfully receives the message with probability one half, and Alice has no information about whether the transmission was successful. From an information-theoretic perspective in the sense of Shannon, this mechanism is equivalent to an erasure channel with erasure probability one half. This conceptualization later inspired the seminal work of Even, Goldreich, and Lempel \cite{EGL}, who formalized OT and established it as a fundamental cryptographic primitive.

It is known that OT cannot be constructed from scratch using cryptographic techniques based solely on indistinguishability assumptions. This impossibility result extends to a wide range of interactive primitives, including bit commitment, zero knowledge, secret sharing, and secure two-party computation, even when adversaries are computationally bounded. Dodis et al. \cite{Dodis} further showed that OT cannot be implemented from any source of imperfect randomness, including Santha–Vazirani (SV) sources \cite{SV} with arbitrarily small bias. Crucially, this impossibility is unconditional and does not rely on computational assumptions or restrictions on adversarial behavior. More generally, any cryptographic functionality that requires privacy is unattainable when only SV-type entropy sources are available. Information theory, therefore, points to noisy communication channels as a principled and viable source of entropy.

Access to noisy channels enables a rich class of information-theoretic cryptographic protocols, a research direction initiated in early foundational works and substantially developed over time \cite{Rabin1, Joe, Winter1, Winter3, Rudolf1, Crepeau3, Imai, Watanabe2, Watanabe3}, and in the context of symmetric private information retrieval (SPIR) has been studied under the terminology OT,
with a noisy channel between parties to achieve information-theoretic security, e.g., \cite{amir1,amir2, amir3}. In contrast, shared noisy channels—where multiple users communicate over a common medium and observe correlated noise realizations—have received comparatively limited attention. Initial investigations appeared only recently \cite{Hadi2, Hadi1, Hadi3, Hadi6, Hadi4}. Unlike independent point-to-point noisy links, shared channels allow participants to leverage noise correlations for coordinated encoding and decoding, reduced communication overhead, and enhanced privacy guarantees. These features are particularly relevant in settings where multiple receivers obtain information from a single source without revealing their individual choices. Consequently, a systematic understanding of shared noisy channels is a key step toward extending OT beyond pairwise communication to more general network models, including broadcast channels (BCs) and multiple-access channels (MACs).

In parallel, an emerging line of research in the context of 6G systems aims to augment classical communication networks with quantum resources, promising improvements in throughput, security, and latency \cite{Granelli}. While substantial progress has been achieved in point-to-point quantum communication \cite{Ekert,Bennett} and in relatively simple multiuser scenarios such as classical MACs \cite{Hsieh, Janis, Leditzky, Seshadri, Pereg1, Fawzi} and BCs \cite{Yard, Pereg2}, more intricate network configurations remain largely unexplored. The work in \cite{Quek} studied rate-splitting advantages enabled by entangled transmitters, introducing an interference channel (IC) motivated by the CHSH game \cite{Clauser} and employing Popescu–Rohrlich box (PR-box)–type superquantum correlations \cite{Popescu}. Subsequent studies \cite{Hawellek1, Hawellek2, Hadi5} analyzed ICs with two senders and two receivers sharing entanglement, demonstrating that quantum advantages induced by non-classical correlations are not confined to MACs but also extend to IC models.

Motivated by these developments, we investigate OT over a discrete memoryless multiple-access channel (DM-MAC) with two senders assisted by a tripartite non-signaling correlation (NS-box). Fawzi et al. \cite{Fawzi} showed that such an NS-box can strictly enlarge the capacity region of a classical DM-MAC. The central question addressed here is whether analogous gains can be achieved for the OT capacity of a DM-MAC \cite{Hadi1}, and, in particular, for the special case of the binary adder channel \cite{Chou}.

This paper is organized as follows: In Section \ref{Sys}, we present the definition of a tripartite NS box in terms of probability distributions and mutual information quantities. We also present OT definitions over both DM-MAC and DM-BCs.
In Section \ref{Results}, we present our main impossibility results of network OT via noisy channels. We also discuss the feasibility of OT via bipartite NS-assisted discrete memoryless point-to-point noisy channels (DMC), showing that OT between two parties via any NS-assisted noisy DMC is not possible if both parties are included in the NS correlation. In Section \ref{Disc}, we have a detailed discussion on the feasibility of network OT over the DM-MAC and DM-BC, where a bipartite NS correlation is shared only between the transmitters or the receivers.

\emph{Notation:} We adopt standard notation from information theory, with the following conventions. Let $\mathbb{N}$ be the set of natural numbers and $\mathbb{R}$ be the set of real numbers.
For any $x, y \in \mathbb{R}$, define $[x:y] \triangleq [\lfloor x \rfloor, \lceil y \rceil] \cap \mathbb{N}
\quad \text{and} \quad
[x] \triangleq [1:x]$.
Script letters such as $\mathcal{X}$ denote finite sets.
Uppercase letters, e.g., $X$, represent classical Random Variables (RVs),
while lowercase letters, e.g., $x$, denote their realizations.
For $n \in \mathbb{N}$, we write $x^n = (x_i)_{i \in [n]}$
for a length-$n$ sequence over $\mathcal{X}$.
The distribution of a discrete RV $X$ is denoted by the Probability Mass Function (PMF)
$p_X(x)$ on the finite alphabet $\mathcal{X}$.
Let $p_X$ and $q_X$ be PMFs on a finite set $\mathcal{X}$.
The total variation distance between $p_X$ and $q_X$ is defined as $\mathbb{V}(p_X, q_X)
\triangleq \frac{1}{2} \sum_{x \in \mathcal{X}} |p_X(x) - q_X(x)|$. We use the notation $\overline{A} \triangleq A \oplus 1$ to denote the binary complement of a bit $A \in \{0,1\}$, where $\oplus$ denotes addition modulo two. To extend this notion to non-binary alphabets, let $\mathcal{I}$ be a finite index set with $|\mathcal{I}| = 2$. Define a bijection $f : \mathcal{I} \to \mathcal{I}$
such that $f(f(i)) = i$ for all $i \in \mathcal{I}$.
We denote $\overline{i} \triangleq f(i)$ and refer to $\overline{i}$ as the complement of $i$.
Hence, $\overline{\overline{i}} = i$, and the mapping $f$ simply swaps the two elements of $\mathcal{I}$. For example, if $i \in \{1, 2\}$, then $\overline{i} = 1$ when $i = 2$.
\section{System Model}\label{Sys}
This section introduces the system models and formal definitions used throughout the paper. We describe network OT over both multiple-access and broadcast channels, with and without non-signaling assistance, and specify the associated correctness and security criteria.
\subsection{Network oblivious transfer over DM-MAC}
We introduce the system model for network OT over a DM-MAC with two senders and one receiver. Our objective is to formalize how OT can be embedded into a shared-channel setting, where multiple transmitters communicate with a common receiver in the presence of noise and public interaction.

We begin by describing the notion of a tripartite NS correlation that may be shared among the two senders and the receiver. This abstraction allows us to model, in a unified manner, both classical channel behavior and additional correlation resources that may arise from quantum or super-quantum mechanisms. We then define the OT protocol, security requirements, and achievable rate notions for the DM-MAC, both with and without NS assistance.

These definitions establish the baseline framework used throughout the paper to analyze the (im)possibility of network OT in MAC settings.
\begin{definition}[Bipartite NS correlation]\label{def: bipartite NS}
Let $\mathcal{I}_1,\mathcal{I}_2,\mathcal{X}_1,\mathcal{X}_2$ be finite sets.
A bipartite NS-box shared between Alice-1 and Alice-2 is specified by a
conditional probability distribution
\begin{align}\label{bipartite-NS}
    P(x_1,x_2 \mid \imath_1,\imath_2),
\end{align}
where $I_1 \in \mathcal{I}_1$ and $I_2 \in \mathcal{I}_2$ are the respective inputs,
and $X_1 \in \mathcal{X}_1$ and $X_2 \in \mathcal{X}_2$ are the respective outputs.
The box is said to be non-signaling if, for every joint input distribution
$P_{I_1I_2}$, the induced random variables $(I_1,I_2,X_1,X_2)$ satisfy
\begin{align}
I(I_1; X_2 \mid I_2) &= 0, \label{bipartite-NS1}\\
I(I_2; X_1 \mid I_1) &= 0. \label{bipartite-NS2}
\end{align}
Equivalently, the conditional distribution satisfies the marginal constraints
\begin{align}
\sum_{x_1} P(x_1,x_2 \mid \imath_1,\imath_2)
&=
\sum_{x_1} P(x_1,x_2 \mid \imath_1',\imath_2),
\quad \forall\, \imath_1,\imath_1',\imath_2,x_2, \label{eq:NS-prob1}\\
\sum_{x_2} P(x_1,x_2 \mid \imath_1,\imath_2)
&=
\sum_{x_2} P(x_1,x_2 \mid \imath_1,\imath_2'),
\quad \forall\, \imath_1,\imath_2,\imath_2',x_1. \label{eq:NS-prob2}
\end{align}
\end{definition}
\begin{definition}[Tripartite NS correlation-assisted DM-MAC \cite{Fawzi}]
A tripartite NS-box over DM-MAC with two senders is described by a joint conditional probability distribution
\begin{equation}\label{NS}
    P(x_1 x_2 ( \jmath_1  \jmath_2) \mid  \imath_1  \imath_2 y),
\end{equation}
where $x_1, x_2$ denote the outputs of the senders, $( \jmath_1,  \jmath_2)$ are the receiver’s outputs,
and $( \imath_1,  \imath_2, y)$ are the respective inputs.  
The NS property requires that the marginal distribution corresponding to any two parties
is independent of the input of the third party. 
That is, for all $x_1, x_2,  \jmath_1,  \jmath_2,  \imath_1,  \imath_2, y,  \imath_1',  \imath_2', y'$, the following conditions hold:
\begin{align}
\sum_{x_1} P(x_1 x_2 ( \jmath_1  \jmath_2) \mid  \imath_1  \imath_2 y) 
&= \sum_{x_1} P(x_1 x_2 ( \jmath_1  \jmath_2) \mid  \imath_1'  \imath_2 y), \label{NS1}\\[6pt]
\sum_{x_2} P(x_1 x_2 ( \jmath_1  \jmath_2) \mid  \imath_1  \imath_2 y) 
&= \sum_{x_2} P(x_1 x_2 ( \jmath_1  \jmath_2) \mid  \imath_1  \imath_2' y), \label{NS2}\\[6pt]
\sum_{ \jmath_1,  \jmath_2} P(x_1 x_2 ( \jmath_1  \jmath_2) \mid  \imath_1  \imath_2 y) 
&= \sum_{ \jmath_1,  \jmath_2} P(x_1 x_2 ( \jmath_1  \jmath_2) \mid  \imath_1  \imath_2 y'). \label{NS3}
\end{align}

Equations~\eqref{NS1}–\eqref{NS3} express that no subset of parties 
can signal information about their inputs to the others through the shared NS resource. Equations \eqref{NS1}-\eqref{NS3} can be written in terms of mutual information quantities as follows: \begin{align}\label{eq: NSM1}
    I(I_1 ; X_2, J \mid I_2, Y) &= 0,\\
    \label{eq: NSM2}
    I(I_2 ; X_1, J \mid I_1, Y) &= 0,\\
    \label{eq: NSM3}
    I(Y ; X_1, X_2 \mid I_1, I_2) &= 0.
\end{align}
\end{definition}
The tripartite NS condition ensures, for example, that the marginal
distribution $ P(x_1 x_2 \mid \imath_1 \imath_2) $ is independent of the 
receiver's input $ y $,
and that $ P(x_1 \mid \imath_1) $ is independent of both $ \imath_2 $ and $ y $.
Hence, in our coding framework, providing NS assistance to the senders and the receiver
amounts to granting them access to a shared tripartite NS-box described by the joint conditional distribution $ P $.

Note first that the property of $P$ being an NS-box is entirely independent of $W$; in particular, the variable $y$ does not need to satisfy any specific distribution in the definition of $P$ as an NS-box. The only remaining subtlety lies in the apparent requirement to perform encoding and decoding simultaneously. However, since $P$ is non-signaling, there is no need to do both at once. Indeed, by the NS property, we have
$$
\forall y, \quad P(x_1 x_2 \mid \imath_1 \imath_2) = P(x_1 x_2 \mid \imath_1 \imath_2 y).
$$
Therefore, the outputs $x_1, x_2$ can be obtained from the inputs $\imath_1, \imath_2$ without requiring knowledge of $y$, as the distribution of $x_1, x_2$ is unaffected by it. Subsequently, the variable $y$ follows the distribution determined by $W$, conditioned on $x_1, x_2$. Finally, given access to $y$, the decoding process is described by:
\begin{equation}\label{eq: causal}
    P((\jmath_1 \jmath_2) \mid \imath_1 \imath_2 y x_1 x_2)
  = \frac{P(x_1 x_2 (\jmath_1 \jmath_2) \mid \imath_1 \imath_2 y)}
         {P(x_1 x_2 \mid \imath_1 \imath_2 y)}
  = \frac{P(x_1 x_2 (\jmath_1 \jmath_2) \mid \imath_1 \imath_2 y)}
         {P(x_1 x_2 \mid \imath_1 \imath_2)}.
\end{equation}
Hence, we recover globally $P((\jmath_1 \jmath_2) \mid \imath_1 \imath_2 y x_1 x_2) \times P(x_1 x_2 \mid \imath_1 \imath_2)
= P(x_1 x_2 (\jmath_1 \jmath_2) \mid \imath_1 \imath_2 y),
$ which corresponds to the prescribed conditional probability. The NS condition guarantees that it is possible to treat the conditional probabilities associated with each party independently. This observation clarifies how one can consistently perform encoding and decoding of messages using an NS-box shared among the three parties. In summary, there is the following causal one-way relation:
\begin{equation*}
(\imath_1,\imath_2)
\xrightarrow[]{\text{NS-box}}
(x_1,x_2)
\xrightarrow[\text{DM-MAC}]{}
y
\xrightarrow[\text{with $(\imath_1,\imath_2)$}]{\text{NS-box}}
(\jmath_1,\jmath_2).
\end{equation*}
\begin{definition}[Nontrivial NS-box-MAC form]\label{def: non-trivial NS}
A tripartite NS-box $P(x_1 x_2 (\jmath_1 \jmath_2) \mid \imath_1 \imath_2 y)$, is said to be \emph{nontrivial} if there exist at least two distinct input pairs
$(\imath_1\imath_2)$ and $(\imath_1',\imath_2')$ and some $y$ such that
\begin{align}\label{non-trivial NS}
    P((\jmath_1 \jmath_2) \mid \imath_1\imath_2y) \neq P((\jmath_1 \jmath_2) \mid \imath_1',\imath_2',y).
\end{align}
Equivalently, the marginal distribution of Bob's outputs depends on the senders' joint inputs.
If this condition does not hold, i.e.,
\[
P((\jmath_1 \jmath_2) \mid \imath_1\imath_2y) = P((\jmath_1 \jmath_2) \mid \imath_1',\imath_2',y)
\quad\text{for all } (\imath_1\imath_2),(\imath_1',\imath_2'),y,
\]
then the box is called \emph{trivial}. A trivial NS-box is one in which Bob's marginal output distribution is independent of the senders' inputs.
Such a box, while satisfying the NS constraints, provides no correlation
beyond independent local randomness.
In particular, for a trivial box, the global distribution factorizes as
\[
P(x_1 x_2 (\jmath_1 \jmath_2) \mid \imath_1 \imath_2 y)
= P(x_1 \mid \imath_1)\, P(x_2 \mid \imath_2)\, P((\jmath_1 \jmath_2) \mid y),
\]
so the outputs of the three parties are statistically independent across locations. Therefore, a trivial NS-box cannot increase the achievable rate region of a multiple-access channel,
since the receiver's outputs are independent of the senders' channel inputs. In other words, the box is said to be \emph{nontrivial} if and only if $I\bigl(I_1,I_2; J \mid Y\bigr) > 0$. Equivalently, the box is \emph{trivial} if and only if $I\bigl(I_1,I_2; J \mid Y\bigr) = 0$,
in which case Bob’s marginal output distribution is independent of the senders’ joint inputs.
\end{definition}

\begin{definition}[OT setup over a DM-MAC, Figure \ref{NS-OT}-$(a)$]\label{def: OT-MAC}
    Let $n, k_1, k_2 \in \mathbb{N} $. An $(n, k_1, k_2)$ protocol involves interaction among Alice-1, Alice-2, and Bob. At each time step $l = 1, 2, \ldots, n $, Alice-$i$ transmits a bit $X_{i,l} $ through the MAC. Users alternately exchange messages over a noiseless public channel in multiple rounds, both prior to each transmission and after the final transmission at $l = n $. While the number of rounds may vary, it is finite. Each user's transmission is determined by a function of their input, private randomness, and all prior messages, channel inputs, or channel outputs observed. A positive rate pair $(R_1, R_2)$ is said to be an achievable OT rate pair for the DM-MAC if for $n\to\infty$ there exist $(n, k_1, k_2)$ protocols satisfying $\frac{k_i}{n} \to R_i$ such that for \emph{non-colluding} parties, the asymptotic conditions \eqref{goals: MAC-nColl-1}--\eqref{goals: MAC-nColl-3} hold:
\begin{align}\label{goals: MAC-nColl-1}
    \lim_{n\rightarrow \infty}\text{Pr}\,\left[(\hat{M}_{1Z_1}, \hat{M}_{2Z_2})\neq (M_{1Z_1}, M_{2Z_2})\right] = 0,\\
    \label{goals: MAC-nColl-2}
    \lim_{n\rightarrow \infty}I (M_{1\overline{Z}_1},M_{2\overline{Z}_2} ; V) = 0,\\
    \label{goals: MAC-nColl-3}
    \lim_{n\rightarrow \infty}I (Z_i ; U_i)_{i \in \{1,2\}} = 0,
\end{align}
where the final view of Alice-$i$ and Bob are $U_i = (M_{i0},M_{i1},R_{A_i},X_i^n, \mathbf{C})$ and $V = (Z_1, Z_2, R_B, Y^n, \mathbf{C})$, respectively, and $\mathbf{C}\triangleq (\mathbf{C}_1, \mathbf{C}_2)$. Conditions \eqref{goals: MAC-nColl-1}-\eqref{goals: MAC-nColl-3} are correctness, security for Alice-$i$ (SfA-$i$), and security for Bob (SfB), respectively. 
\end{definition}

\begin{definition}[Network OT setup via tripartite NS correlation-assisted DM-MAC, Figure \ref{NS-OT}-$(c)$]
    Consider a DM-MAC with two senders, Alice-1 and Alice-2, and one receiver, Bob. The parties are additionally provided with a shared tripartite NS resource characterized by the conditional distribution defined in \eqref{NS} with the properties as stated in \eqref{NS1}-\eqref{NS3}. An $(n,k_1,k_2)$ \emph{NS assisted OT protocol} over the DM-MAC is defined as an $(n,k_1,k_2)$ OT protocol in which the parties are additionally allowed to make use of the shared NS resource described above. At each time step $l=1,\dots,n$, Alice-$i$ produces a channel input $X_{i,l}$ and possibly an NS-box input $\imath_{i,l}$,  according to her local view of the protocol and previous outputs of the NS-box. Bob may also input a variable $y_l$ to the NS-box and obtain corresponding outputs $(\jmath_{1,l}, \jmath_{2,l})$, which may be used in subsequent channel or local computations. A rate pair $(R_1,R_2)$ is said to be \emph{achievable under NS assistance} if, for $n\to\infty$, 
there exist $(n,k_1,k_2)$ NS-assisted protocols satisfying conditions \eqref{goals: MAC-nColl-1}-\eqref{goals: MAC-nColl-3},
where $U_i = (M_{i0}, M_{i1}, R_{A_i}, X_i^n, \mathbf{C})$ denotes the final view of Alice-$i$, 
and $V = (Z_1, Z_2, R_B, Y^n, \mathbf{C})$ the final view of Bob.
Here $R_{A_i}$ and $R_B$ represent local randomness, $\mathbf{C}$ denotes public communication, 
and the random variables may depend on both the MAC and the NS-box.

\medskip
The closure of all achievable rate pairs $(R_1, R_2)$ satisfying 
\eqref{goals: MAC-nColl-1}--\eqref{goals: MAC-nColl-3} is called the \emph{NS assisted OT capacity region} of the MAC with non-colluding parties.

\end{definition}
\begin{figure*}[tb]

\includegraphics[scale=0.859,trim={2.47cm 32cm 3cm 3.5cm},clip]{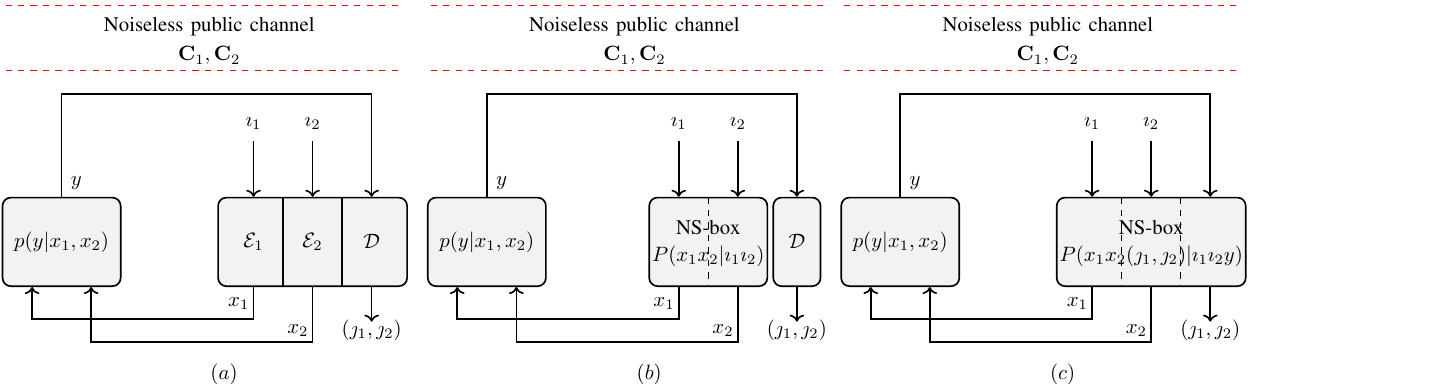} %

\caption{$(a)$ The OT system model over a two-user noisy DM-MAC with two encoders and a decoder. $(b)$ The OT system model over a bipartite NS correlation-assisted noisy DM-MAC. $(c)$ The OT system model over a tripartite NS correlation-assisted noisy DM-MAC.}
\label{NS-OT}
\end{figure*}
\begin{definition}[Perfect security criteria in terms of NS-box]
    Fix a blocklength $n$. For $i\in\{1,2\}$ denote by\linebreak $U_i = \bigl(M_{i0},M_{i1},R_{A_i},X_i^n,A_i^n,\mathbf{C}\bigr)$ the entire local view of Alice-$i$ at the end of the protocol: her two files, local randomness $R_{A_i}$, the channel inputs $X_i^n$ (or the NS-inputs $\imath_i^n$), her local NS-box outputs $A_i^n$ (the sequence of $a_{i,t}$), and the public transcript $\mathbf{C}$. Bob's selection bits are $Z=(Z_1,Z_2)$ and his full view is $V=(Z,Y^n,J,R_B,\mathbf{C})$ as before.
    \begin{align}\label{eq: Bob sec-NS}
    P(U_i=u_i\mid M=m,Z=z) &= P(U_i=u_i\mid M=m,Z=z')\quad\text{for all }u_i\\
    \label{eq: Alice sec-NS}
    P\bigl(V=v \mid \imath_1\imath_2,Z=z\bigr)
    &=
    P\bigl(V=v \mid \imath'_1\imath'_2,Z=z\bigr) \quad\text{for all } v. 
    \end{align}
\end{definition}
\subsection{Network OT over DM-BC}
We now consider the dual network setting of a DM-BC, in which a single sender communicates with multiple receivers. In this model, the sender holds multiple message pairs, while each receiver privately selects one message to be delivered without revealing its choice to the sender or to the other receiver.

As in the MAC case, we allow the parties to share a tripartite NS correlation, which may influence the encoding, decoding, and channel interaction while respecting the NS constraints. We first define the structure of such NS correlations for the broadcast scenario and then formalize the corresponding notions of OT, correctness, and security.

The resulting framework enables a direct comparison between MAC and BC settings and serves as the foundation for the impossibility results established in Section \ref{Results}.
\begin{definition}[Tripartite NS-box over DM-BC]
A tripartite NS-box over DM-BC with two receivers is described by a joint conditional probability distribution
\begin{equation}\label{NS-BC}
    P(x \jmath_1  \jmath_2 \mid  \imath y_1 y_2),
\end{equation}
where $x$ denotes the output of the sender, $ \jmath_i$ is the receiver-$i$’s output,
and $(\imath, y_1, y_2)$ are the respective inputs.  
The NS property requires that the marginal distribution corresponding to any two parties
is independent of the input of the third party. 
That is, for all $x,  \jmath_1,  \jmath_2,  \imath, y_1, y_2,  \imath', y'_1,y'_2$, the following conditions hold:
\begin{align}
\sum_{x} P(x \jmath_1  \jmath_2 \mid  \imath y_1 y_2) 
&= \sum_{x} P(x \jmath_1  \jmath_2 \mid  \imath' y_1 y_2), \label{NS1-BC}\\
\sum_{\jmath_1} P(x \jmath_1  \jmath_2 \mid  \imath y_1 y_2) 
& = \sum_{\jmath_1} P(x \jmath_1  \jmath_2 \mid  \imath y'_1 y_2), \label{NS2_BC}\\
\sum_{\jmath_2} P(x \jmath_1  \jmath_2 \mid  \imath y_1 y_2) 
& = \sum_{ \jmath_2} P(x \jmath_1  \jmath_2 \mid  \imath y_1 y'_2). \label{NS3-BC}
\end{align}

The above NS conditions imply that the marginal distribution of any two parties is independent of the input of the third. In particular, the sender-side output can be generated without knowledge of the receivers’ inputs, since for all $(y_1,y_2)$,
\[
P(x \mid \imath)=P(x \mid \imath y_1 y_2).
\]
As a consequence, encoding and decoding need not be performed simultaneously. Operationally, this induces the following one-way causal structure:
\[
\imath
\xrightarrow[]{\text{NS-box}}
x
\xrightarrow[\text{DM-BC}]{}
(y_1,y_2)
\xrightarrow[\text{with $\imath$}]{\text{NS-box}}
(j_1,j_2).
\]
That is, the sender-side output $x$ is produced solely from the sender’s input $\imath$, the broadcast channel generates the receiver-side observations $(y_1,y_2)$, and finally the receiver-side NS-box outputs $(j_1,j_2)$ are generated conditioned on $(\imath,y_1,y_2)$ in a manner consistent with the joint conditional distribution. This causal separation guarantees that the prescribed NS-box distribution can be realized while respecting the NS constraints.

Equations~\eqref{NS1-BC}–\eqref{NS3-BC} express that no subset of parties 
can signal information about their inputs to the others through the shared NS resource. Equations \eqref{NS1-BC}-\eqref{NS3-BC} can be written in terms of mutual information quantities as follows: \begin{align}\label{eq: NSM1-BC}
    I(I ; J_1,J_2 \mid  Y_1,Y_2) &= 0,\\
    \label{eq: NSM2-BC}
    I(Y_1 ; X, J_2 \mid I, Y_2) &= 0,\\
    \label{eq: NSM3-BC}
    I(Y_2 ; X, J_1 \mid I, Y_1) &= 0.
\end{align}
\end{definition}
\begin{definition}[Non-trivial NS-Box-BC form]\label{def: NS-BC-triv}

We consider a broadcast channel with conditional distribution $p(y_1,y_2 \mid x)$, and a tripartite NS-box specified by the conditional distribution $P(xj_1j_2 \mid \imath y_1 y_2)$,
where $I$ denotes the sender input to the NS-box, $(Y_1,Y_2)$ are side-information inputs, $X$ is the output fed into the broadcast channel, and $(J_1,J_2)$ are the receivers' outputs. The induced joint distribution is
\[
P(\imath,x,y_1,y_2,j_1,j_2)
=
P(\imath)\,P(x,j_1,j_2 \mid \imath,y_1,y_2)\,p(y_1,y_2\mid x).
\]
The tripartite NS-box is said to be \emph{trivial} if
\begin{equation}
P(x,j_1,j_2 \mid \imath,y_1,y_2)
=
P(x,j_1,j_2 \mid y_1,y_2),
\quad \forall\, \imath,y_1,y_2.
\label{eq:prob_trivial}
\end{equation}
The tripartite NS-box is said to be \emph{non-trivial} if there exist
$\imath\neq \imath'$ and $y_1,y_2,x,j_1,j_2$ such that
\begin{equation}
P(x,j_1,j_2 \mid \imath,y_1,y_2)
\neq
P(x,j_1,j_2 \mid \imath',y_1,y_2).
\label{eq:prob_nontrivial}
\end{equation}
In other words, the tripartite NS-box is \emph{trivial} if and only if
\begin{equation}
I\!\left(I \,;\, X,J_1,J_2 \mid Y_1,Y_2\right) = 0.
\label{eq:mi_trivial}
\end{equation}
Equivalently,
\[
H(X,J_1,J_2 \mid Y_1,Y_2)
=
H(X,J_1,J_2 \mid I,Y_1,Y_2).
\]
\end{definition}
\begin{definition}\label{def: main-BC}
   Let $n, k_1, k_2 \in \mathbb{N}$. An $(n, k_1, k_2)$ protocol is an interactive communication among Alice, Bob-1, and Bob-2. Over $n$-time steps, Alice transmits a single bit $X_l$ over the broadcast channel at each time $l = 1, \ldots, n$. In addition, the parties may exchange messages over a noiseless public channel in a finite number of rounds both before each transmission and after the final transmission. Each message sent by a party may depend on its input, private randomness, and all previously observed messages, channel inputs, and channel outputs. A rate pair $(R_1, R_2)$ is said to be achievable for oblivious transfer over the DM-BC if, as $n \to \infty$, there exist $(n, k_1, k_2)$ protocols with $\frac{k_i}{n} \to R_i$ such that the asymptotic non-collusion conditions \eqref{goals: BC-nColl-1}--\eqref{goals: BC-nColl-3} are satisfied.
   \begin{align}\label{goals: BC-nColl-1}
    \lim_{n\rightarrow\infty}\text{Pr}\,\left[(\hat{M}_{iZ_i})\neq (M_{iZ_i})\right]_{i\in\{1,2\}}= 0,\\
    \label{goals: BC-nColl-2}
    \lim_{n\rightarrow\infty}I (  M_{i\overline{Z}_i};V_i)_{i\in\{1,2\}} = 0,\\
    \label{goals: BC-nColl-3}
    \lim_{n\rightarrow\infty}I (Z_1, Z_2 ; U) =  0,
\end{align}
where $V_i = (Z_i, R_{B_i}, Y_i^n, \mathbf{C}_i), U = (M_{i0}, M_{i1}, R_{A}, X^n , \mathbf{C}), i\in \{1,2\}$
and $\mathbf{C}\triangleq (\mathbf{C}_1, \mathbf{C}_2)$. Conditions \eqref{goals: BC-nColl-1}-\eqref{goals: BC-nColl-3} are correctness, security for Alice (SfA), and security for Bob-$i$ (SfB-$i$), respectively.
\end{definition}
\begin{definition}[Network OT setup via tripartite NS assisted DM-BC-Figure \ref{fig: NS-OT-BC}-$(b)$]\label{def: OT_NS_BC}
    Consider a DM-BC with two receivers. The parties are additionally provided with a shared tripartite NS resource characterized by the conditional distribution defined in \eqref{NS-BC} with the properties as stated in \eqref{NS1-BC}-\eqref{NS3-BC}. An $(n,k_1,k_2)$ \emph{NS assisted OT protocol} over the DM-BC is defined as an $(n,k_1,k_2)$ OT protocol in which the parties are additionally allowed to make use of the shared NS resource described above. At each time step $l=1,\dots,n$, Alice produces a channel input $X_{l}$ and possibly an NS-box input $\imath_{l}$,  according to her local view of the protocol and previous outputs of the NS-box. Bob-$i$ may also input a variable $y_{i,l}$ to the NS-box and obtain corresponding outputs $\jmath_{i,l}$, which may be used in subsequent channel or local computations. A rate pair $(R_1,R_2)$ is said to be \emph{achievable under NS assistance} if, for $n\to\infty$, there exist $(n,k_1,k_2)$ NS-assisted protocols satisfying conditions \eqref{goals: BC-nColl-1}-\eqref{goals: BC-nColl-3}, where $U = (M_{10}, M_{11}, M_{20}, M_{21}, R_{A}, X^n, \mathbf{C})$ denotes the final view of Alice, and $V_i = (Z_i, R_{B_i}, Y_i^n, \mathbf{C})$ the final view of Bob-$i$. Here $R_{A}$ and $R_{B_i}$ represent local randomness, $\mathbf{C}$ denotes public communication, and the random variables may depend on both the DM-BC and the NS-box.
    The closure of all achievable rate pairs $(R_1, R_2)$ satisfying 
\eqref{goals: BC-nColl-1}--\eqref{goals: BC-nColl-3} is called the \emph{NS assisted OT capacity region} of the MAC with non-colluding parties.

\end{definition}
\begin{figure*}[tb]

\includegraphics[scale=0.9,trim={0.2cm 19.75cm 3cm 1.8cm},clip]{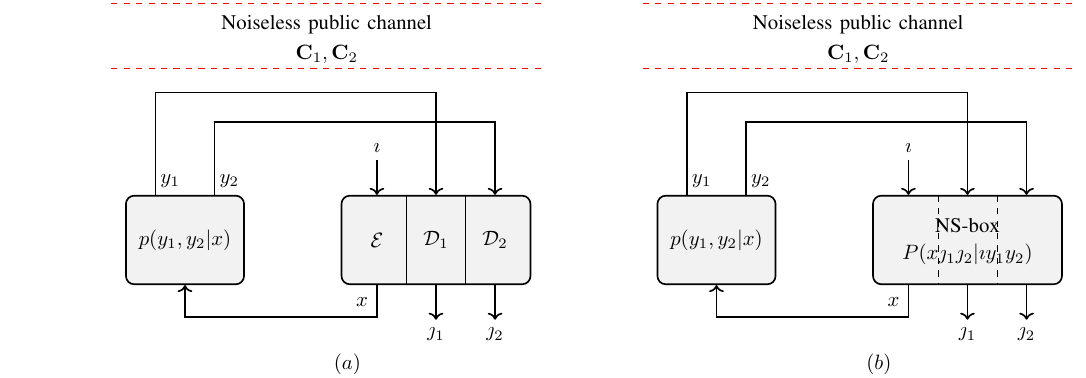}

\caption{(a) The OT system model over a noisy DM-BC with two decoders. (b) The OT system model over a noisy DM-BC with an NS-box.}
\label{fig: NS-OT-BC}
\end{figure*}

\section{Main Results and Proofs}\label{Results}
\subsection{Perfect OT over DM-MAC under a nontrivial tripartite NS-box}
We begin by studying the feasibility of perfect oblivious transfer over a two-sender DM-MAC when all parties are assisted by a shared tripartite NS correlation. The goal of this subsection is to understand whether the presence of such a powerful correlation resource can circumvent the classical impossibility of OT over shared noisy channels.

We consider the strongest possible assistance scenario, in which the NS-box is shared among both senders and the receiver and may be arbitrarily nontrivial, while the encoders are deterministic functions of the senders’ files. Under these assumptions, we show that perfect OT cannot be achieved: any protocol that enables correct decoding of the selected messages necessarily violates the privacy of at least one sender.

The following theorem formalizes this impossibility result.
\begin{theorem}
Consider a two-sender DM-MAC: $ W : \mathcal{X}_1 \times \mathcal{X}_2 \to \mathcal{Y} $. Let the NS-box inputs be deterministic functions\footnote{Note that $f_i$ is not an extra encoder on top of the NS-box.
They just model how each party chooses its box input $\imath_i ,i\in\{1,2\}$ as a deterministic function of its messages.} $ \imath_1 = f_1 (M_{10},M_{11}) \in \{0,1\}^{n}, \imath_2 = f_2 (M_{20},M_{21}) \in\{0,1\}^{n} $, respectively,
and Bob observes the channel output $ y \in \mathcal{Y} $.
Suppose that all three parties share a tripartite NS correlation
$ P(x_1, x_2, (\jmath_1 \jmath_2) \mid \imath_1, \imath_2, y) $,  
satisfying the standard NS constraints \eqref{NS1}-\eqref{NS3}, then OT is impossible under such an NS assistance.
\end{theorem}

\begin{proof}
By the definition of the NS-box, the decoding law of Bob is
given by the conditional \eqref{eq: causal}:
\[
P((\jmath_1 \jmath_2) \mid \imath_1 \imath_2 y x_1 x_2)
  = \frac{P(x_1 x_2 (\jmath_1 \jmath_2) \mid \imath_1 \imath_2 y)}
         {P(x_1 x_2 \mid \imath_1 \imath_2)}.
\]
Also, we have,
\begin{align}
    P(x_1 x_2(\jmath_1\jmath_2)\,y \mid \imath_1 \imath_2)=P(x_1 x_2(\jmath_1\jmath_2) \mid \imath_1 \imath_2 y) \,W(y\mid x_1,x_2).
\end{align}
By the definition of conditional probability and marginalization,
\begin{align}
    P(y\mid \imath_1\imath_2)
    &= \sum_{x_1x_2\jmath_1\jmath_2} P\bigl(x_1x_2(\jmath_1\jmath_2)y \mid \imath_1\imath_2\bigr) \notag\\
    &=
    \sum_{x_1x_2\jmath_1\jmath_2}
    P\bigl(x_1x_2(\jmath_1\jmath_2)\mid \imath_1\imath_2\,y\bigr)\, W(y\mid x_1,x_2).
\end{align}
This conditional law generally depends jointly on both inputs $ \imath_1 $
and $ \imath_2 $, because the NS condition only ensures
invariance after marginalization over one party’s outputs.

\smallskip
However, in a secure OT protocol, the privacy constraint
requires that Bob’s view of the channel output be independent of each
unselected sender’s message:
\begin{align}\label{eq: leak1}
P(y \mid \imath_1, \imath_2) &= P(y \mid \imath_1', \imath_2), \quad
  \forall \imath_1,\imath_1',\imath_2,\\
  \label{eq: leak2}
P(y \mid \imath_1, \imath_2) &= P(y \mid \imath_1, \imath_2'), \quad
  \forall \imath_2,\imath_2',\imath_1. 
\end{align}
But by the general structure of the NS correlation-assisted MAC,
there exist $ (\imath_1\imath_2,\imath_1',\imath_2',y) $ such that
\[
P(y \mid \imath_1, \imath_2) \neq P(y \mid \imath_1', \imath_2),
\]
since the decoding law depends jointly on $ (\imath_1, \imath_2) $
through the ratio
$ P(x_1 x_2 (\jmath_1 \jmath_2) \mid \imath_1 \imath_2 y) / P(x_1 x_2 \mid \imath_1 \imath_2) $.
This implies that Bob’s decoding output is statistically correlated with
both senders’ inputs simultaneously. Consequently, the receiver can infer information about both senders’
messages, violating the OT privacy conditions \eqref{eq: leak1}–\eqref{eq: leak2}.
The contradiction shows that even in the presence of a
tripartite NS correlation between all parties,
no protocol can achieve OT over a MAC channel. As we stated before, for $i\in\{1,2\}$ Alice-$i$ has two files $M_{i0},M_{i1}$.
Bob chooses selection bits $z = (z_1,z_2)\in\{0,1\}^2$.
Each sender uses a (deterministic) encoder $\imath_i = f_i(M_{i0},M_{i1})$,
which is the input they feed to the NS-box (or the MAC coding index).
Let $V$ denote Bob's entire \emph{final} view (this may include the MAC outputs $Y$, NS-box outputs $J=(\jmath_1,\jmath_2)$, public transcript $\mathbf{C}$, and Bob's local randomness).  For a single round, we denote Bob's observable from the NS assisted MAC by $v=(y,j)$; in the $n$-round case, $V = (Y^n, J, Z, \mathbf{C})$.

For OT secrecy we require that, for each sender $i$ and for every fixed choice of the selection bits $Z_1,Z_2$, Bob's view is independent of the \emph{unselected} file $M_{i\overline{Z}_i}$ when the \emph{selected} file $M_{iZ_i}$ is held fixed.  Formally, for every $i\in\{1,2\}$, every choice bits vector $z=(z_1,z_2)$, and every two message tuples
\[
m=(m_{10},m_{11},m_{20},m_{21}),\qquad m'=(m'_{10},m'_{11},m'_{20},m'_{21}),
\]
such that
\[
M_{iz_i} = m_{iz_i},\quad\text{(the chosen file for Alice-$i$ is equal)}
\]
and
\[
m_{k\ell}=m'_{k\ell}\quad\text{for all } (k,\ell)\neq (i,\overline{Z}_i)
\quad\text{(all other files equal)},
\]
The senders' secrecy \eqref{eq: Alice sec-NS} requires
\begin{equation}\label{eq:message-secrecy}
P\bigl(V=v \mid M=m,\; Z=z\bigr)
\;=\;
P\bigl(V=v \mid M=m',\; Z=z\bigr)
\quad\text{for all } v.
\end{equation}
This means that, when the chosen files and all other files coincide, changing the unselected file must not affect Bob's view. Because the encoders are deterministic, $\imath_i=f_i(M_{i0},M_{i1})$, the condition \eqref{eq:message-secrecy} is equivalent to the following constraint on the induced conditional distributions of Bob's observables:
for every $z$ and for every pair of encoder inputs $(\imath_1\imath_2)$ and $(\imath'_1,\imath'_2)$ that arise from message tuples $m,m'$ satisfying the identicality conditions above (same chosen files, all other files equal), we must have
\begin{equation}\label{eq:encoder-secrecy}
P\bigl(V=v \mid \imath_1\imath_2,Z=z\bigr)
\;=\;
P\bigl(V=v \mid \imath'_1\imath'_2,Z=z\bigr)
\quad\text{for all } v.
\end{equation}
The left-hand side can be written explicitly via the NS-box and MAC composition; e.g. for a single round with MAC law $W(y\mid x_1,x_2)$ and NS kernel $P(x_1x_2(\jmath_1\jmath_2)\mid \imath_1\imath_2y)$,

For fixed sender inputs $(\imath_1\imath_2)$,
the joint behaviour of the NS-box and the MAC is given by
\begin{equation}\label{eq:joint}
P(x_1x_2(\jmath_1\jmath_2),y\mid \imath_1\imath_2)
= P(x_1x_2(\jmath_1\jmath_2)\mid \imath_1\imath_2y)\, W(y\mid x_1,x_2).
\end{equation}
Hence Bob’s marginal distribution on the observable pair $(y,(\jmath_1\jmath_2))$ is
\begin{equation}\label{eq:marginal}
P(y,(\jmath_1\jmath_2)\mid \imath_1\imath_2)
= \sum_{x_1,x_2}
P(x_1x_2(\jmath_1\jmath_2)\mid \imath_1\imath_2y)\, W(y\mid x_1,x_2).
\end{equation}

Including his choice $z$ and possible public transcript $\mathbf{C}$, overall distribution of Bob's full view is therefore
\begin{align}\label{eq:fullview}
    P(V=v\mid \imath_1\imath_2)
    & = \delta_{Z=z}\;
    P(\mathbf{C}=c\mid \imath_1\imath_2,z)\;
    P(y,(\jmath_1\jmath_2)\mid \imath_1\imath_2)\notag\\
    &= 
    \delta_{Z=z}\;
    P(\mathbf{C}=c\mid \imath_1\imath_2,z)\;
    \sum_{x_1,x_2}
    P(x_1x_2(\jmath_1\jmath_2)\mid \imath_1\imath_2y)\, W(y\mid x_1,x_2),
\end{align}
where $\delta_{Z=z}$ simply fixes the choice bits in Bob’s view to the selected value.
From the nontriviality of the NS-box (Definition \ref{def: non-trivial NS}), this violates Alices' secrecy \eqref{eq:encoder-secrecy}.
\end{proof}

We have the following proposition that states that only trivial NS-boxes can not violate OT secrecy via DM-MACs.
\begin{proposition}\label{prop1}
Let $P(x_1x_2(\jmath_1\jmath_2)\mid \imath_1\imath_2y)$ be a tripartite NS kernel
and let the MAC be $W(y\mid x_1,x_2)$.  Consider the induced distribution of Bob's view
$V=(Z,Y,(\jmath_1\jmath_2),\mathbf{C})$ conditioned on the encoder inputs $(\imath_1\imath_2)$ and Bob's choice $Z=z$.
If the following universal encoder-level secrecy holds:
\begin{equation}\label{eq: NS Alice's secrecy}
\forall\ \imath_1,\imath_2,\imath_1',\imath_2',\ \forall v,\qquad
P(V=v\mid \imath_1\imath_2,Z=z)=P(V=v\mid\imath_1'\imath_2',Z=z),
\end{equation}
then Bob's marginal distribution is independent of the senders' inputs; equivalently, the NS-assisted MAC resource is \emph{trivial}. Conversely, if the NS-assisted MAC is nontrivial, then \eqref{eq: NS Alice's secrecy} is false.
\end{proposition}

\begin{proof}
Assume \eqref{eq: NS Alice's secrecy} holds. Marginalizing (or simply restricting to the observable part), we obtain, for all $\imath_1,\imath_2,\imath_1',\imath_2'$ and for every $y,(\jmath_1\jmath_2)$,
\[
P(y,(\jmath_1\jmath_2)\mid \imath_1\imath_2,Z=z)
= P(y,(\jmath_1\jmath_2)\mid \imath_1',\imath_2',Z=z).
\]
Dropping the irrelevant conditioning on $Z=z$ (it is fixed in the equality) we write
\[
P(y,(\jmath_1\jmath_2)\mid \imath_1\imath_2)
= P(y,(\jmath_1\jmath_2)\mid \imath_1',\imath_2').
\]
Hence the marginal distribution of Bob's observable $(Y,(\jmath_1\jmath_2))$ is the same for every encoder-input pair; i.e.
\begin{align}\label{eq: 1}
    P(y,(\jmath_1\jmath_2)\mid \imath_1\imath_2) \equiv Q(y,(\jmath_1\jmath_2))
\quad\text{(independent of }\imath_1\imath_2).
\end{align}

Using the composition (factorization) of the NS kernel with the MAC,
\[
P(x_1,x_2,(\jmath_1\jmath_2),y\mid \imath_1\imath_2)
= P(x_1,x_2,(\jmath_1\jmath_2)\mid \imath_1\imath_2y)\;W(y\mid x_1,x_2),
\]
and marginalizing over $x_1,x_2$ we obtain
\begin{align}\label{eq: 2}
    Q(y,(\jmath_1\jmath_2))
= \sum_{x_1,x_2} P(x_1,x_2,(\jmath_1\jmath_2)\mid \imath_1\imath_2y)\;W(y\mid x_1,x_2),
\end{align}
for every $(\imath_1\imath_2)$.

Now \eqref{eq: 2} says that the right-hand side does not depend on $(\imath_1\imath_2)$. In particular,
the marginal distribution of $Y$ given $(\imath_1\imath_2)$ is independent of $(\imath_1\imath_2)$:
\[
P(Y=y\mid \imath_1\imath_2) = \sum_{\jmath_1,\jmath_2} Q(y,(\jmath_1\jmath_2)) \quad\text{(independent of $\imath$)}.
\]
Thus Bob's marginal $P(Y\mid \imath_1\imath_2)$ is constant in $(\imath_1\imath_2)$. Similarly,
$P((\jmath_1\jmath_2)\mid \imath_1\imath_2y)$ must be such that the sum \eqref{eq: 2} is independent of $(\imath_1\imath_2)$.

This is exactly the definition of a \emph{trivial} resource in our context: Bob's distribution does not depend on the senders' inputs, hence the NS-assisted MAC provides no dependence/correlation useful for distinguishing different encoder inputs. Conversely, if the NS-assisted MAC is nontrivial — i.e. there exist $(\imath_1\imath_2)\neq(\imath_1',\imath_2')$ with
\[
P(y,(\jmath_1\jmath_2)\mid \imath_1\imath_2)\neq P(y,(\jmath_1\jmath_2)\mid \imath_1'\imath_2'),
\]
then \eqref{eq: NS Alice's secrecy} is false.
\end{proof}
\begin{remark}
Proposition \ref{prop1} characterizes when Bob’s entire view is independent of the senders’ encoder inputs, namely when the NS-assisted MAC resource is trivial.  
However, this condition is strictly stronger than the secrecy requirement of oblivious transfer. In an OT protocol, secrecy is conditional in terms of reliability: Bob must not obtain information about the \emph{unselected} messages, while still being able to reliably decode the selected ones.  
In contrast, the universal encoder-level secrecy \eqref{eq: NS Alice's secrecy} requires Bob’s view to be independent of \emph{all} encoder inputs, which precludes any reliable communication. Consequently, if the NS-assisted MAC is trivial, then correctness fails, and OT is impossible.  
If the NS-assisted MAC is nontrivial, then \eqref{eq: NS Alice's secrecy} necessarily fails, implying that Bob can distinguish some encoder inputs, which violates perfect SfA. Thus, perfect OT is impossible regardless of whether the tripartite NS assistance is trivial or nontrivial.  
Proposition \ref{prop1} should therefore be understood as a characterization of triviality, not as an achievability result for OT.
\end{remark}

\begin{proposition}
Consider a two-sender DM-MAC: $ W : \mathcal{X}_1 \times \mathcal{X}_2 \to \mathcal{Y} $. Let the NS-box inputs be deterministic functions$ \imath_1 = f_1 (M_{10},M_{11}) \in \{0,1\}^{n}, \imath_2 = f_2 (M_{20},M_{21}) \in\{0,1\}^{n} $, respectively,
and Bob observes the channel output $ y \in \mathcal{Y} $.
Suppose that all three parties share a non-trivial/trivial tripartite NS correlation
$ P(x_1, x_2, (\jmath_1 \jmath_2) \mid \imath_1, \imath_2, y) $,
satisfying the standard NS constraints \eqref{eq: NSM1}-\eqref{eq: NSM3}, then OT is impossible under such an NS assistance.
\end{proposition}

\begin{proof}
    Consider \eqref{goals: MAC-nColl-2}:
    \begin{align}
            I (M_{1\overline{Z}_1},M_{2\overline{Z}_2} ; V) & \,= I (M_{1\overline{Z}_1},M_{2\overline{Z}_2} ; Z, Y, J, \mathbf{C})\notag\\
            & \,= I (M_{1\overline{Z}_1} ; Z, Y, J, \mathbf{C})\notag + I (M_{2\overline{Z}_2} ; Z, Y, J, \mathbf{C}|M_{1\overline{Z}_1})\notag\\
            & \overset{(a)}{=} I (M_{1\overline{Z}_1} ; Z, Y, J, \mathbf{C})\notag + I (M_{2\overline{Z}_2} ; Z, Y, J, \mathbf{C}, M_{1\overline{Z}_1})\notag\\
            & \overset{(b)}{=} I (M_{1\overline{Z}_1} ; Z, Y, J, \mathbf{C}) + I (M_{2\overline{Z}_2} ; Z, Y, J, \mathbf{C})\notag\\
            & = I (M_{1\overline{Z}_1} ; V)+I (M_{2\overline{Z}_2} ; V)\notag\\
            & = 0,\label{MAC-nColl-middle}
        \end{align}
        where $(a)$ is due to the independency of $M_{1\overline{Z}_1}$ from $M_{2\overline{Z}_2}$, and $(b)$ is due the Markov chain $M_{2\overline{Z}_2}-(Z,Y, J, \mathbf{C})-M_{1\overline{Z}_1}$. Then, $I (M_{1\overline{Z}_1} ; V) = 0$.
    
    Since the NS correlation does not have any impact on the public channel transmission, then,  we ignore $\mathbf{C}$ as the total public transmission transcript. Consider the first NS-box criterion \eqref{eq: NSM1}, $I(I_1;X_2,J|I_2,Y)=0$:
    \begin{align}\label{eq:midd}
        I(I_1;X_2,J|I_2,Y) & = I(I_1;J|I_2,Y)\notag\\
        & \quad + I(I_1;X_2|I_2,Y,J)\notag\\
        & = 0.
    \end{align}
Knowing that $V=(J,Z,Y^n)$, we have:
\begin{align}
    I(I_1;V|I_2,Y) &= I(I_1;J,Z,Y|I_2,Y)\notag\\
    & = I(I_1;J|I_2,Y)\notag\\
    & \quad+ I(I_1;Z|I_2,Y,J)\notag\\
    & \quad+ I(I_1;Y|I_2,Y,J,Z)\notag\\
    & = 0,
\end{align}
where the last equality follows from \eqref{eq:midd} and the fact that the last two terms are equal to zero.
Data processing from non-signaling and deterministic encoding forms the following Markov chain: $(M_{1Z_1}, M_{1\overline{Z}_1})-I_1-I_2,Y-V$.
since $I_1$ is a deterministic function of $(M_{1Z_1},M_{1\overline Z_1})$, the
NS condition implies
\[
I(M_{1Z_1},M_{1\overline Z_1};V\mid I_2,Y)=0,
\]
or in entropy form
\begin{equation}\label{eq:1}
    H(M_{1Z_1},M_{1\overline Z_1}\mid I_2,Y) = H(M_{1Z_1},M_{1\overline Z_1}\mid V,I_2,Y).
\end{equation}
Chain-rule expansions of \eqref{eq:1}:
\begin{align}\label{eq:1'}
    H(M_{1Z_1}\mid I_2,Y)+H(M_{1\overline Z_1}\mid M_{1Z_1},I_2,Y)
=
H(M_{1Z_1}\mid V,I_2,Y)+H(M_{1\overline Z_1}\mid M_{1Z_1},V,I_2,Y).
\end{align}

\textit{(i) Zero-error transmission:} Correctness implies $H(M_{1Z_1}\mid V,I_2,Y)=0$.
Then, \eqref{eq:1'} implies:
\begin{align}
    H(M_{1Z_1}\mid I_2,Y)+H(M_{1\overline Z_1}\mid M_{1Z_1},I_2,Y) &= H(M_{1\overline Z_1}\mid M_{1Z_1},V,I_2,Y)\notag\\ 
    & \leq H(M_{1\overline Z_1}\mid M_{1Z_1},I_2,Y),
\end{align}
hence
\begin{equation}\label{eq:1''}
    H(M_{1Z_1}\mid I_2,Y)=0.
\end{equation}
This shows the selected message $M_{1Z_1}$ is already determined by $(I_2,Y)$,
contradicting nontrivial Alice-1 security unless messages are degenerate.

If $H(M_{1Z_1}\mid I_2,Y)=0$, 
then nontrivial Alice-1 secrecy
\[
I(M_{1\overline Z_1};V)=0,
\]
is violated, except in the degenerate case where the non-selected message has zero entropy (i.e. is deterministic). Consider \eqref{eq:1'}:
\begin{equation*}
    H(M_{1Z_1}\mid I_2,Y)+H(M_{1\overline Z_1}\mid M_{1Z_1},I_2,Y) = H(M_{1Z_1}\mid V,I_2,Y)+H(M_{1\overline Z_1}\mid M_{1Z_1},V,I_2,Y).
\end{equation*}
Since $H(M_{1Z_1}\mid V,I_2,Y)=0$, we obtain:
\begin{align}\label{eq:2}
    H(M_{1Z_1}\mid I_2,Y)+H(M_{1\overline Z_1}\mid M_{1Z_1},I_2,Y) = H(M_{1\overline Z_1}\mid M_{1Z_1},V,I_2,Y).
\end{align}
\noindent From \eqref{eq:1''} we have $H(M_{1Z_1}\mid I_2,Y)=0$, Then:
\begin{align}\label{eq:2'}
    H(M_{1\overline Z_1}\mid M_{1Z_1},I_2,Y)=H(M_{1\overline Z_1}\mid M_{1Z_1},V,I_2,Y).
\end{align}
On the other hand, conditioning does not increase entropy,
\[
H(M_{1\overline Z_1}\mid M_{1Z_1},V,I_2,Y)\le H(M_{1\overline Z_1}\mid M_{1Z_1},I_2,Y),
\]
and \eqref{eq:2'} therefore forces equality in this monotonicity. Equality in conditional-entropy monotonicity implies that $V$ (equivalently the part of $V$ beyond $(I_2,Y,M_{1Z_1})$) carries no additional information about $M_{1\overline Z_1}$. Since $Y\in V$, a direct consequence is that, whenever $(I_2,Y)$ fixes $M_{1Z_1}$ (which it does by \eqref{eq:1''}, the remaining uncertainty of $M_{1\overline Z_1}$ given $(I_2,Y)$ equals the remaining uncertainty given $(I_2,Y,V)$. On the other hand, non-triviality of the NS-box implies that $I(I_1,I_2;J|Y)>0$ and $I_i = f_i (M_{i0}, M_{i1}), i\in\{1,2\}$. Considering $V = (J,Z,Y^n)$, we have:
\begin{align}
    I(I_1,I_2;J|Y) & \,= I(I_1;J|Y) + I(I_1;J|I_2,Y)\notag\\
    & \stackrel{(a)}{=} I(I_1;J|Y)\notag\\
    \label{eq: midd'}
    & \stackrel{(b)}{\leq} I(I_1;V|Y)\\
    & \,>0,
\end{align}
where $(a)$ follows from \eqref{eq: NSM1} and $(b)$ follows from data processing inequality. Consider \eqref{eq: midd'}:
\begin{align}\label{eq: midd''}
    I(I_1;V|Y) & \stackrel{(a)}{=} I(M_{1Z_1},M_{1\overline{Z}_1};V|Y)\notag\\
    & \,= H(M_{1Z_1},M_{1\overline{Z}_1}|Y) - H(M_{1Z_1},M_{1\overline{Z}_1}|V, Y)\notag\\
    & \,=H(M_{1Z_1}|Y) + H(M_{1\overline{Z}_1}|M_{1Z_1}, Y) - H(M_{1Z_1}|V, Y) - H(M_{1\overline{Z}_1}|M_{1Z_1}, V, Y)\notag\\
    & \,>0,
\end{align}
where $(a)$ follows from this fact that $I_1$ is a deterministic function of $M_{10}, M_{11}$. Since conditioning does not increase entropy, $H(M_{1\overline{Z}_1}|M_{1Z_1}, Y) \geq H(M_{1\overline{Z}_1}|M_{1Z_1},V, Y)$, combined with \eqref{eq: midd''}, we have $H(M_{1\overline{Z}_1}|M_{1Z_1}, Y) > H(M_{1\overline{Z}_1}|M_{1Z_1},V, Y)$. However, this contradicts \eqref{eq:2'}, meaning that, conditioning on $V$ leaks information about Alice’s unselected message, and Alice-1 security is violated. 

\textit{(ii) Asymptotic correctness:} Consider \eqref{goals: MAC-nColl-1}. Applying Fanno's inequality gives:
\begin{align}\label{cor-MAC-asym}
    H(M_{1Z_1},M_{2Z_2}|V)\leq h(\varepsilon_n) + \varepsilon_n \log_2|\mathcal{M}|,
\end{align}
where $M\triangleq (M_{1Z_1},M_{2Z_2})\in\{0,1\}^{k_1+k_2}$, and $\log_2|\mathcal{M}| = \log_2 2^{k_1+k_2} = k_1+k_2$. Then, \eqref{cor-MAC-asym} can be written as:
\begin{align}\label{cor-MAC-asym-2}
    H(M_{1Z_1},M_{2Z_2}|V)\leq \eta_n,
\end{align}
where $\eta_n \triangleq h(\varepsilon_n) + \varepsilon_n (k_1+k_2)$. 
Substituting \eqref{cor-MAC-asym-2} into \eqref{eq:1'}, gives:
\begin{align}\label{eq:1'''}
    H(M_{1Z_1}\mid I_2,Y)+H(M_{1\overline Z_1}\mid M_{1Z_1},I_2,Y) \leq H(M_{1\overline Z_1}\mid M_{1Z_1},V,I_2,Y) + \eta_n.
\end{align}
On the other hand,
$
H(M_{1\overline Z_1}\mid M_{1Z_1},V,I_2,Y)
\leq
H(M_{1\overline Z_1}\mid M_{1Z_1},I_2,Y).
$
If equality holds in this monotonicity (i.e., condition SfA-1 is satisfied), then
$
H(M_{1Z_1}\mid I_2,Y_1)\leq \eta_n,
$
where $\eta_n \to 0$ as $n \to \infty$. This violates the causality\footnote{OT is a \emph{causal} cryptographic primitive in the sense that the receiver obtains a message determined by a locally generated choice variable. Accordingly, the receiver’s final reconstructed message depends on this choice, and the causality is reflected in the delivered message $M_{iZ_i}$, $i\in\{1,2\}$. The causal dependence appears at the level of the final decoding: the reconstructed message is a function of the receiver’s view and his choice. In entropy terms, this implies $H(M_{iZ_i}\mid V_i,Z_i)=0$, while no analogous identity holds for intermediate variables such as the NS-box output $J_i$.} of OT. If the inequality is strict, then Alice-1's perfect secrecy is violated. The argument for Alice-2 follows analogously.
This completes the proof.
\end{proof}

\subsection{Asymptotic OT over DM-MAC under a nontrivial tripartite NS-box}
While the impossibility result of the previous subsection rules out perfect OT, one may still hope to achieve OT with asymptotically vanishing error and leakage by allowing blocklength to grow. In particular, it is natural to ask whether repeated use of a nontrivial tripartite NS-assisted MAC could enable asymptotic OT with negligible information leakage.

In this subsection, we show that this is not the case. Even when reliability is required only asymptotically and the protocol is allowed to operate over many independent channel uses, any nontrivial tripartite NS-box induces a distinguishability between encoder inputs that amplifies under repetition. As a consequence, the receiver can asymptotically learn information about unselected messages.

The main result of this subsection establishes the impossibility of asymptotic OT under these conditions.
\begin{theorem}[Impossibility of asymptotic OT with a nontrivial tripartite NS-box]
Assume the parties share a tripartite NS-box $P(x_1x_2(\jmath_1\jmath_2)\mid\imath_1,\imath_2,y)$ that is nontrivial in the sense above, the encoders $\imath_i=f_i(M_{i0},M_{i1})$ are deterministic functions of the senders' files, and the parties use independent repetitions of the same single-use kernel (product across rounds). Then there exists no sequence of protocols that simultaneously satisfies reliability and the senders-security asymptotically.
\end{theorem}

\begin{proof}
Assume there exists a sequence of $(n, k_1, k_2)$-protocols that satisfy reliability and Alice-security \eqref{goals: MAC-nColl-2}. We will construct two message tuples $m,m'$ that differ only in an unselected file and show Bob can distinguish them with probability tending to $1$, contradicting Alice-security.

By nontriviality there exist inputs $(\imath_1^*,\imath_2^*)\neq(\tilde\imath_1,\tilde\imath_2)$ and $y^*$ such that Bob's single-round output distributions  differ by (total variation distance)
\[
\varepsilon := \mathbb{V}\bigl(P_{J\mid \imath_1^*,\imath_2^*,y^*},\; P_{J\mid \tilde\imath_1,\tilde\imath_2,y^*}\bigr) > 0.
\]
(Here $J$ denotes Bob's view for the single round; if the MAC output $Y$ is also observed by Bob we include it in $J$ or treat conditioning on $y^*$ as in the model; the statement is unchanged up to incorporating $Y$ into the observable.)

Because the encoders are deterministic mappings $\imath_i = f_i(M_{i0},M_{i1})$, the senders can choose message tuples that map to any desired pair of encoder inputs in their ranges. In particular, pick two message tuples
\[
m=(m_{10},m_{11},m_{20},m_{21}),\qquad m'=(m'_{10},m'_{11},m'_{20},m'_{21}),
\]
such that
\[
f_1(m_{10},m_{11}) = \imath_1^*, \quad f_2(m_{20},m_{21})=\imath_2^*,
\]
and
\[
f_1(m'_{10},m'_{11})=\tilde\imath_1, \quad f_2(m'_{20},m'_{21})=\tilde\imath_2.
\]
(If encoder ranges do not include the desired inputs, the nontrivial pair can be chosen among achievable inputs; the argument only requires the existence of two achievable input pairs that the box distinguishes.)

Moreover we choose $m,m'$ so that they agree on all chosen files for some fixed choice $z$ but differ in an unselected file, so that OT secrecy would require Bob not to be able to distinguish them. Concretely, pick $z$ and suppose $m,m'$ differ only in $M_{1\overline z_1}$ (or $M_{2\overline z_2}$); the argument is identical in either case. By construction, under messages $m$ and $m'$ the encoder inputs to the NS-box are $(\imath_1^*,\imath_2^*)$ and $(\tilde\imath_1,\tilde\imath_2)$ respectively, therefore Bob's single-round observable distributions differ by $\varepsilon$ in total variation:
\[
\mathbb{V}\bigl(P_{J\mid m},P_{J\mid m'}\bigr)\;=\;\varepsilon.
\]

Since the protocol uses independent repetitions (product across rounds), the $n$-round distributions $P_{J^n\mid m}$ and $P_{J^n\mid m'}$ are $n$-fold product distributions of the single-round ones (or at least their marginals on the $J$-components factorize similarly). For product distributions the total variation distance tensorizes as
\[
\mathbb{V}\bigl(P^{\otimes n},Q^{\otimes n}\bigr) \;=\; 1 - (1-\mathbb{V}(P,Q))^n.
\]
Consequently,
\[
\mathbb{V}\bigl(P_{J^n\mid m},P_{J^n\mid m'}\bigr)
= 1 - (1-\varepsilon)^n.
\]
In particular $\mathbb{V}(\cdot)\to 1$ exponentially fast as $n\to\infty$.

Consider the binary hypothesis testing problem where Bob must decide which of the two message tuples $m$ or $m'$ was used, based on observing $J^n$. With equal a priori probabilities for $m$ and $m'$, the minimal probability of error $P_e^{(n)}$ is given by
\[
P_e^{(n)} \;=\; \frac{1}{2}\bigl(1 - \mathbb{V}(P_{J^n\mid m},P_{J^n\mid m'})\bigr).
\]
Equivalently, the maximal probability of correct decision is
\[
p_{\mathrm{corr}}^{(n)} \;=\; 1 - P_e^{(n)} \;=\; \frac{1}{2}\bigl(1 + \mathbb{V}(P_{J^n\mid m},P_{J^n\mid m'})\bigr)
= 1 - \tfrac12 (1-\varepsilon)^n.
\]
Thus $p_{\mathrm{corr}}^{(n)}\to 1$ as $n\to\infty$, and $P_e^{(n)}\to 0$.

Define the binary random variable $B$ that equals $0$ if the message tuple is $m$ and $1$ if it is $m'$, with equal prior probability $1/2$. By the standard relation between binary hypothesis-testing error and mutual information for an equiprobable binary $B$ and observation $J^n$,
\[
I(B;J^n) = H(B) - H(B\mid J^n) = 1 - H_2(P_e^{(n)}),
\]
where $H_2(\cdot)$ is the binary entropy in bits and $H(B)=1$ bit. Since $P_e^{(n)}\to 0$, we have $H_2(P_e^{(n)})\to 0$, and therefore
\[
\lim_{n\to\infty}I(B;J^n) = 1.
\]
Because $B$ is a deterministic function of the unselected file(s) (we chose $m,m'$ so they differ only in an unselected file), we have the data processing lower bound
\[
I\bigl(M_{1\overline Z_1},M_{2\overline Z_2}; J^n\bigr) \;\ge\; I(B;J^n)\xrightarrow[n\to\infty]{} 1.
\]
Hence the mutual information between the unselected file(s) and Bob's $n$-round view is bounded below by a positive constant (in fact at least approaching 1 bit), which contradicts the Alice-security requirement that
\[
\lim_{n\to\infty} I\bigl(M_{1\overline Z_1},M_{2\overline Z_2}; V^{(n)}\bigr) = 0.
\]
(Here $V^{(n)}$ includes $J^n$ and possibly more information available to Bob, so $I(\cdot;J^n)\le I(\cdot;V^{(n)})$; the same contradiction therefore holds even more strongly for $V^{(n)}$.)

Any nontrivial tripartite NS-box gives rise to a single-round distinguishability $\varepsilon>0$ between certain encoder inputs; because the encoders map messages to inputs deterministically, one can arrange two message tuples that differ only in an unselected file but produce distinguishable Bob distributions; repetition amplifies this distinguishability so Bob learns the unselected file asymptotically with probability tending to one. Thus Alice-security (asymptotic vanishing leakage) is impossible. This completes the proof.
\end{proof}

\subsection{On Bob's security in the presence of a tripartite NS-box}
The impossibility results above focus on the violation of sender security. This naturally raises the complementary question of whether SfB is also subject to a universal impossibility when tripartite non-signaling assistance is available.

In contrast to sender security, we show that no such universal limitation exists for Bob’s security. Depending on how public communication and NS-box outputs are used within the protocol, Bob’s selection bits may or may not be leaked. Thus, SfB is a protocol-dependent property rather than an inherent limitation imposed by the NS-assisted channel model.

The following theorem illustrates this distinction.

\begin{theorem}[No universal impossibility for Bob-security] There is no universal impossibility result on Bob's security in an OT scheme in the presence of a nontrivial tripartite NS-box: Bob-security can hold for some protocols and fail for others.
\end{theorem}

\begin{proof}
Consider the following setting:
\begin{enumerate}
  \item All public transcript messages $\mathbf{C}$ are generated exclusively as functions of the senders' local randomness and messages (and previous public messages), but \emph{not} of $Z$ or of any Bob information;
  \item Encoders $\imath_i^n=f_i^{(n)}(M_{i0},M_{i1})$ are deterministic functions of the senders' files only and never depend on $Z$ (i.e.\ senders never adapt to Bob);
  \item No additional side channels (no MAC feedback used adaptively by senders that depends on $Z$).
\end{enumerate}

For the $n$-round protocol, denote the global conditional distribution
given the messages $M=m$ and Bob's choice $Z=z$ by
\begin{align}
\label{eq:global-dist}
P_{\mathrm{global}}(&x_1^n,x_2^n,j^n,y^n,c \mid M=m,Z=z) = \notag\\
  &\Biggl[\prod_{t=1}^n
      P\bigl(x_{1t},x_{2t},(\jmath_{1t}\jmath_{2t})\mid
      \imath_{1t},\imath_{2t},y_t\bigr)\,
      W(y_t\mid x_{1t},x_{2t})
  \Biggr]
  P(\mathbf{C}=c\mid \text{protocol inputs}),
\end{align}
where $\mathbf{C}$ denotes the public transcript random variable,
and $\imath_i^n=f_i^{(n)}(M_{i0},M_{i1})$ are the deterministic encoder inputs. The local view of Alice-$i$ over $n$ rounds is $U_i^{(n)}=(M_{i0},M_{i1},\imath_i^n,X_i^n,\mathbf{C})$ where $X_i^n=(x_{i1},\ldots,x_{in})$ is the sequence of her MAC inputs.

\vspace{3pt}
To obtain the marginal distribution of Alice-$i$'s view conditioned on $(M,Z)$,
we must sum over all variables that she cannot observe.
These are precisely the variables of the other sender, Bob,
and the channel outputs, i.e. $\text{hidden variables:}\quad x_{\overline{i}}^n,\ j^n,\ y^n$.

The full expression for Alice-$i$'s conditional marginal therefore reads:
\begin{align}
\label{eq:alice-marginal-explicit}
P(U_i^{(n)}&=u_i \mid M=m, Z=z) =
\nonumber\\
&\sum_{x_{\overline{i}}^n,\jmath^n,y^n}
   \Biggl[
   \prod_{t=1}^n
     P\bigl(x_{1t},x_{2t},(\jmath_{1t}\jmath_{2t})\mid
     \imath_{1t},\imath_{2t},y_t\bigr)
     W(y_t\mid x_{1t},x_{2t})
   \Biggr]\notag\\
   &\qquad\qquad\qquad\qquad\qquad\qquad\qquad\qquad\qquad P(\mathbf{C}=c\mid M=m,Z=z,x_1^n,x_2^n,j^n,y^n).
\end{align}

If the transcript generation law $P(\mathbf{C}\mid\cdot)$ is independent of $Z$ and the NS marginal equalities hold for
$P$, then the right-hand side of
\eqref{eq:alice-marginal-explicit} is independent of $z$,
implying Bob-security for Alice-$i, i\in\{1,2\}$.

Next, let the protocol be identical to the previous one except that, in round 1, after receiving his local NS output $(\jmath_1,\jmath_2)$ (or after computing a deterministic function of $y$), Bob broadcasts a single bit like
\[
B = g(\jmath_1,\jmath_2,y,Z),
\]
where $g$ is chosen so that $B$ depends on $Z$ (for example $B=Z_1$, or $B$ is a hash revealing part of $Z$). This obviously violates Bob's security.
\end{proof}

\subsection{Perfect OT over DM-BC under a nontrivial tripartite NS-box}
We now turn to the DM-BC setting and investigate whether analogous impossibility results hold for network OT over DM-BCs in the presence of tripartite non-signaling correlations.
\begin{theorem}
Consider a DM-BC, $ W : \mathcal{X} \to \mathcal{Y}_1\times \mathcal{Y}_2$. Let the NS-box inputs be deterministic functions $ \imath = f (M_{10},M_{11},M_{20},M_{21}) \in \{0,1\}^{n}$,
and Bob-$i$ observes the channel output $ y_i \in \mathcal{Y}_i, i\in\{1,2\}$.
Suppose that all three parties share a tripartite non-trivial/trivial NS correlation
$ P(x, \jmath_1, \jmath_2 \mid \imath, y_1, y_2) $,
satisfying the standard NS constraints \eqref{eq: NSM1-BC}-\eqref{eq: NSM3-BC}, then OT (Definition \ref{def: OT_NS_BC}) is impossible under such an NS assistance.
\end{theorem}
\begin{proof}
    Consider the non-signaling from Alice to both receivers expressed in \eqref{eq: NSM1-BC}:
    \begin{align}
        I(I ; J_1,J_2 \mid  Y_1,Y_2) & = I(I ; J_1\mid  Y_1,Y_2) + I(I ; J_2 \mid  J_1,Y_1,Y_2)\notag\\
        & = 0.
    \end{align}
    Then $I(I ; J_1\mid  Y_1,Y_2) = 0$.
    Next, we have the following chain of equalities:
    \begin{align}
        I(I;V_1|Y_1,Y_2) &= I(I;J_1,Z_1,Y_1|Y_1,Y_2)\notag\\
        & = I(I;J_1|Y_1,Y_2) + I(I;Z_1|Y_1,Y_2, J_1) + I(I;Y_1|Y_1,Y_2,J_1,Z_1)\notag\\
        & = 0,
    \end{align}
    where the last equality follows immidiately since $Z_1\perp\!\!\!\! \perp (I, Y_1, Y_2, J_1)$ meaning that Bob-1's choice bit is locally generated and is independent of 
    Alice's messages ($I-(Y_1, Y_2, J_1)-Z_1$).

    Since $I$ is a deterministic function of all messages, we have:
    \begin{align}
        I(I;V_1|Y_1,Y_2) &= I(M_{10},M_{11},M_{20},M_{21};V_1|Y_1,Y_2)\notag\\
        & = I(M_{1Z_1},M_{1\overline{Z}_1},M_{2Z_2},M_{2\overline{Z}_2};V_1|Y_1,Y_2)\notag\\
        & = I(M_{1Z_1},M_{1\overline{Z}_1};V_1|Y_1,Y_2) + I(M_{2Z_2},M_{2\overline{Z}_2};V_1|Y_1,Y_2, M_{1Z_1},M_{1\overline{Z}_1})\notag\\
        & = 0.
    \end{align}
    Then, $I(M_{1Z_1},M_{1\overline{Z}_1};V_1|Y_1,Y_2) = 0$, and we have:
    \begin{align}
        H(M_{1Z_1},M_{1\overline{Z}_1}|Y_1,Y_2) = H(M_{1Z_1},M_{1\overline{Z}_1}|Y_1,Y_2,V_1).
    \end{align}
Expanding both sides implies:
    \begin{align}\label{eq: doubt}
        H(M_{1Z_1}|Y_1,Y_2) + H(M_{1\overline{Z}_1}|Y_1,Y_2, M_{1Z_1}) = H(M_{1Z_1}|Y_1,Y_2,V_1) + H(M_{1\overline{Z}_1}|Y_1,Y_2,V_1, M_{1Z_1}).
    \end{align}
    Correctness (zero-error transmission) of OT requires $H(M_{1Z_1}|V_1) = 0$. Then, $H(M_{1Z_1}|Y_1,Y_2,V_1) = 0$ since conditioning does not increase entropy. Then:
    \begin{align}\label{eq: after-corr}
        H(M_{1Z_1}|Y_1,Y_2) + H(M_{1\overline{Z}_1}|Y_1,Y_2,M_{1Z_1}) = H(M_{1\overline{Z}_1}|Y_1,Y_2,V_1, M_{1Z_1}).
    \end{align}
    On the other hand $H(M_{1\overline{Z}_1}|Y_1,Y_2,V_1, M_{1Z_1}) \leq H(M_{1\overline{Z}_1}|Y_1,Y_2, M_{1Z_1})$. If the equality holds in this monotonicity (satisfying SfA), then \eqref{eq: after-corr} requires $H(M_{1Z_1}|Y_1,Y_2) = 0$, meaning that the chosen message by Bob-1 is completely determined from the BC output and is independent from his final view $V_1$, including his choice bit $Z_1$. This contradicts the causality of OT. If the inequality holds in the monotonicity, then it violates \eqref{eq: after-corr}. 

    In the above proof, we considered zero-error transmission over the noisy BC, a very strict assumption. Now, consider the asymptotic correctness constraint \eqref{goals: BC-nColl-1} $\lim_{n\to\infty}\text{Pr}\{\Hat{M}_{1Z_1} \neq M_{1Z_1}\} = 0$. In other form $\text{Pr}\{\Hat{M}_{1Z_1} \neq M_{1Z_1}\} = \varepsilon_n$, where $\varepsilon_n\to 0$ as $n\to\infty$. Then, Fano's inequality gives:
    \begin{align}
        H(M_{1Z_1}|V_1)\leq h(\varepsilon_n) + \varepsilon_n\log_2|\mathcal{M}_{1Z_1}|.
    \end{align}
    Since $M_{1Z
    _1}\in\{0,1\}^{k_1}$ and $k_1$ is the message length, then:
    \begin{align}\label{ineq}
        H(M_{1Z_1}|V_1)\leq \eta_n,
    \end{align}
    where $\eta_n \triangleq h(\varepsilon_n) + \varepsilon_nk_1$. Now, start with \eqref{eq: doubt}:
    \begin{align}\label{eq: doubt2}
        H(M_{1Z_1}|Y_1,Y_2) + H(M_{1\overline{Z}_1}|Y_1,Y_2, M_{1Z_1}) = H(M_{1Z_1}|Y_1,Y_2,V_1) + H(M_{1\overline{Z}_1}|Y_1,Y_2,V_1, M_{1Z_1}).
    \end{align}
    Substituting \eqref{ineq}, we have:
    \begin{align}\label{eq: doubt3}
        H(M_{1Z_1}|Y_1,Y_2) + H(M_{1\overline{Z}_1}|Y_1,Y_2, M_{1Z_1}) \leq H(M_{1\overline{Z}_1}|Y_1,Y_2,V_1, M_{1Z_1}) + \eta_n.
    \end{align}
    On the other hand $H(M_{1\overline{Z}_1}|Y_1,Y_2,V_1, M_{1Z_1}) \leq H(M_{1\overline{Z}_1}|Y_1,Y_2, M_{1Z_1})$. If the equality holds in this monotonicity (satisfying SfA), then $H(M_{1Z_1}|Y_1,Y_2)\leq\eta_n$, and $\eta_n$ vanishes when $n\to\infty$. This contradicts the causality of OT. If the inequality holds, then it contradicts Alice's perfect secrecy. The proof for Bob-2 follows immediately. 
    This completes the proof.
\end{proof}

\subsection{On impossibility of OT over a bipartite NS-assisted noisy point-to-point DMC}We have shown that there exists no sequence of OT protocols that simultaneously satisfy reliability and the SfA over tripartite NS-assisted DM-MAC and tripartite NS-assisted DM-BC. This constitutes a strong no-go result, as it is completely independent of the noise behavior and noise type of the underlying channels. 

Moreover, this impossibility remains valid beyond the honest-but-curious model and extends to stricter secrecy notions, including bounded malicious parties \cite{Winter1} and unbounded malicious parties \cite{Hadi4}. In this sense, our result establishes an impossibility that is strictly stronger than existing computational impossibility results, since OT is ruled out even in the presence of powerful NS correlations.

Fawzi and Ferm\'e \cite{Fawzi} showed that bipartite NS correlation cannot increase either the capacity of a DMC or the sum-capacity of a DM-MAC. Technically, their result applies to the setting in which NS assistance is shared independently between Alice-1 and Bob, described by $P^1(x_1,\jmath_1 \mid \imath_1,y)$, as well as independently between Alice-2 and Bob, described by $P^2(x_2,\jmath_2 \mid \imath_2,y)$, which we refer to as independent NS assistance. From the perspective of network OT, our impossibility results can be extended to bipartite NS-assisted noisy DMCs and independent NS-assisted DM-MACs. However, the (im)possibility of weaker interactive cryptographic tasks, such as bit-commitment and secret key agreement, under such NS assistance remains an open problem.
\subsection{On OT over a bipartite NS-assisted noisy DM-MAC}
Up to now, we have examined the feasibility of network OT over noisy channels assisted by tripartite NS correlations. We have shown that neither perfect nor imperfect OT can be realized in this setting. In particular, any tripartite NS correlation either violates the SfA, leaks unintended information to the receiver, or contradicts the causal structure inherent to OT.

This leads to a natural and central question: can OT be realized over shared noisy channels assisted by bipartite NS correlations, where the NS resource is shared exclusively between the two transmitters in the MAC model (Figure \ref{NS-OT}-$(b)$) or, alternatively, between the two receivers in the BC model? In this work, we restrict our attention to the MAC model and investigate the case in which the two transmitters are assisted by a bipartite NS correlation.  

\begin{theorem}
\label{thm:bipartiteNS-no-impossibility}
Consider a noisy DM-MAC,
$W:\mathcal X_1\times\mathcal X_2\to\mathcal Y$.
Assume Alice-$1$ and Alice-$2$ share a bipartite non-signaling (NS) box
$P_{X_1X_2|I_1I_2}$ satisfying \eqref{bipartite-NS1} and \eqref{bipartite-NS2},
while the receiver Bob has no access to the NS box.
Then, no universal information-theoretic impossibility of perfect OT
over the DM-MAC can be derived.
\end{theorem}
\begin{proof}
Let $M_i=(M_{i0},M_{i1})$, $i\in\{1,2\}$, be independent and uniformly distributed over $\{0,1\}^{k_i}$,
and let $Z=(Z_1,Z_2)$ denote Bob’s choice bits.
Let $V=(Y^n,\mathbf{C},Z)$ be Bob’s full view. Perfect OT correctness/zero-error transmission implies
\begin{align}
H(M_{1Z_1},M_{2Z_2}\mid V)=0,
\end{align}
and hence $I(M_{1Z_1},M_{2Z_2};Y^n\mid Z)=H(M_{1Z_1},M_{2Z_2})$. Perfect SfA-$i$ requires $I(M_{i\bar Z_i};V)=0, i=\{1,2\}$.

In the tripartite NS-assisted case, a contradiction follows from
non-signaling constraints that imply
$I(I_i;Y^n,J\mid I_{\bar i})=0$, where $J$ is Bob’s NS output.
Here, however, the only NS constraints are \eqref{bipartite-NS1} and \eqref{bipartite-NS2}, which involve sender-side variables exclusively. Since Bob observes neither $(X_1,X_2)$ nor any NS-box output,
no constraint enforces $I(I_i;Y^n\mid I_{\bar i})=0$.
Consequently, the Markov relations required to deduce
\begin{align}
H(M_{iZ_i}\mid Y^n,Z,I_{\bar i})=0,
\end{align}
from SfA-$i$ cannot be established.

Therefore, the entropy identities used in Proposition \ref{prop1} to force
leakage of $M_{i\bar Z_i}$ do not follow. However, if the parties are allowed to collude, then each sender can reveal the private information to Bob, and this collapses OT requirements. 
\end{proof}
\subsection{NS correlations and the collapse of choice-dependent causality in OT}
Oblivious transfer is inherently a \emph{causal} cryptographic primitive: the receiver’s final reconstructed message depends essentially on a locally generated choice variable. Specifically, Bob-$i$ obtains his estimate $\hat M_{iZ_i}$ as a deterministic function of his total view $V_i$, which includes the outputs of the NS-box and the public transcript, together with his private choice $Z_i$. Formally, $H(\hat M_{iZ_i} \mid V_i) = 0$,
while $\hat M_{iZ_i}$ cannot be reconstructed from $V_i\setminus\{Z_i\} \triangleq (J_i, \mathbf{C}, Y_i^n)$. In an asymptotically correct protocol, $\hat M_{iZ_i}$ coincides with the true message $M_{iZ_i}$ with vanishing error probability, so that $H(M_{iZ_i} \mid V_i) \to 0$.
The causal role of the choice variable thus manifests at the level of the final decoded message, rather than in the raw outputs of the NS-box.

This entropy relation captures a fundamental aspect of OT causality: the receiver’s output cannot be fully determined without knowledge of the choice bit. In other words, the uncertainty in $\hat{M}_{iZ_i}$ necessarily reflects the uncertainty in $Z_i$, and the decoding process must preserve this dependence. Any valid OT protocol must therefore ensure that the receiver’s final output remains informationally and causally tied to the choice variable.

NS correlations fundamentally undermine this causal structure. NS constraints enforce that, conditioned on the channel outputs, the receiver’s entire view—including the final output—cannot carry additional information about the sender’s private input. As a result, any information that enables the receiver to recover the selected message must already be contained in the channel outputs alone. This forces the decoding of $M_{iZ_i}$ to become asymptotically independent of the choice variable $Z_i$, thereby eliminating the causal influence of the choice on the final output.

From an entropic perspective, NS constraints effectively push the system toward a regime in which $H(M_{iZ_i} \mid Y_1,Y_2) \to 0$ (DM-BC),
so that the receiver’s output becomes asymptotically determined by the channel observations. In this regime, the entropy contribution of the choice variable disappears, violating the fundamental OT requirement that the output depends on the choice. To conclude, NS correlations erase the causal pathway through which the receiver’s choice influences the delivered message. While NS correlations can introduce strong nonlocal correlations, they simultaneously enforce a form of informational symmetry that prevents the asymmetric, choice-dependent information flow required by OT. This causal collapse lies at the heart of the impossibility results established in this work.
\subsection{OT from nonlocal primitives versus NS assistance}

A closely related and influential result was established by Wolf and Wullschleger \cite{Wolf}, who study the cryptographic power of the PR non-local box. In their model, two parties have access to a bipartite non-signaling primitive in which each party inputs a single bit and receives a uniformly random output bit, such that the XOR of the outputs equals the AND of the inputs. The central contribution of that work is to show an exact and information-theoretic equivalence between OT and the PR primitive: a single use of a PR-box suffices to implement a perfectly secure instance of 1-out-of-2 OT, and, conversely, a single instance of OT can be used to simulate a PR-box. Both reductions are single-copy, perfect (zero-error), and information-theoretic, and therefore establish the PR-box as a complete primitive for OT in the NS setting. In the PR-based constructions of their work, the PR-box is treated as an abstract bipartite cryptographic primitive, and OT is obtained by a direct reduction using local operations and classical message exchanges between the parties. The model is not embedded into a physical communication network and does not impose any encoder–channel–decoder structure.

At first glance, our impossibility result may appear to contradict earlier work establishing a cryptographic equivalence between OT and PR nonlocality, where it is shown that a single realization of a PR primitive suffices to implement information-theoretically secure OT, and conversely that OT can be used to simulate PR correlations. This apparent tension is resolved by observing that the two settings operate under fundamentally different models of resources and causality. In the PR-based constructions, the PR-box is treated as an abstract bipartite cryptographic primitive whose input--output correlations are directly exploited to generate the OT functionality; the reduction is stand-alone, single-copy, and is carried out in a setting with an otherwise noiseless public channel, without embedding the primitive into a physical communication network. In particular, the PR-box is used as a nonlocal primitive that provides the essential cryptographic correlation, and is not constrained by any intermediate noisy channel or operational encoding--decoding structure. By contrast, in our model, NS correlations are available only as auxiliary resources assisting communication over a physically constrained channel, such as a DM-MAC or DM-BC. All actual information flow is required to pass through the noisy channel, while the NS-box merely provides correlations that cannot convey information by themselves. Consequently, the protocol is subject to a strict operational causal ordering---encoding, channel transmission, and decoding---which is absent in the abstract PR-box framework. Our impossibility results exploit precisely this operational causality: the receiver’s choice in OT must causally influence the decoded message, while non-signaling constraints forbid any backward influence of this choice on the sender-side views. This tension does not arise in PR-based reductions, where the primitive itself already embodies the required nonlocal dependence. Thus, our result does not contradict the PR--OT equivalence, but rather shows that NS correlations, when confined to an assisting role within a causally ordered communication protocol, are insufficient to universally guarantee OT security. This distinction highlights a fundamental difference between nonlocal primitives used as stand-alone cryptographic functionalities and NS resources employed as auxiliary correlations in communication networks.

Other closely related conceptual limitations of NS resources were identified by Short, Gisin, and Popescu in their study of the cryptographic power of nonlocal correlations \cite{Short} and by Buhrman \emph{et al.} \cite{Buhrman}, where they show how non-local
boxes can be used to perform any two-party secure computation. In \cite{Short}, the authors investigate the operational differences between an ideal OT-box and a PR nonlocal box, and show that reproducing the same input--output probability distribution is, in general, insufficient to reproduce the cryptographic functionality of a primitive. In particular, they emphasize that, besides matching the joint distribution, one must also respect the temporal and causal structure of the inputs and outputs of the ideal functionality. For an OT-box, the receiver’s output is only defined after the sender has entered her inputs, whereas for any NS-box, and in particular for a PR-box, the receiver must be able to obtain a valid output independently of whether the other party has already used the box. This fundamental difference implies that, although an OT-box can be simulated from a PR-box and classical communication at the level of distributions, the simulation does not preserve the causal ordering inherent in the OT functionality and is therefore not universally composable. As summarized by Short, Gisin, and Popescu in the slogan ``no signaling---no obligation to perform the measurement until the end---no commitment'', the non-signaling constraint prevents a protocol from enforcing that a party has already fixed a causally relevant local input at an earlier stage of the protocol. Our model and results reflect precisely the same obstruction in a communication-theoretic setting. In the NS-assisted channel model considered here, the NS-box is available only as an auxiliary resource, and all information must flow through a physically constrained broadcast or multiple-access channel. Consequently, the protocol is subject to an explicit operational ordering of encoding, channel transmission, and decoding. The non-signaling constraints then imply that the auxiliary resource cannot enforce that a party’s local choice has already been causally applied before the other parties’ outputs are generated. Hence, although NS correlations may reproduce certain input--output statistics, they cannot, in general, reproduce the causal structure required by OT over a communication network. In this sense, the impossibility results of \cite{Short} for nonlocal simulations provide an operational and composability-based confirmation of the limitations established in our NS-assisted channel model.

The apparent tension between the possibility results of Wolf and Wullschleger and the impossibility results of Short, Gisin, and Popescu is resolved once the role of causality is made explicit. Wolf and Wullschleger establish a functional equivalence between OT and PR nonlocality at the level of input--output correlations: a PR-box can simulate the distribution of an OT-box, and vice versa, without imposing any operational constraints on the timing of the parties’ inputs. In contrast, Short, Gisin, and Popescu show that NS correlations cannot enforce the temporal binding required by commitment-type primitives, since the NS condition forces a party’s output to be well defined independently of whether the other party has already acted. In our model, OT is embedded into a causally ordered communication protocol—encoding, physical channel transmission, and decoding—and therefore implicitly requires a causal dependence between private choices and final outputs. NS assistance, while capable of generating correlations, cannot impose this causal ordering and thus cannot universally guarantee OT security over communication channels. From this perspective, the possibility result of Wolf and Wullschleger and the impossibility result of Short, Gisin, and Popescu are fully compatible: the former concerns the simulation of ideal functionalities at the level of correlations, whereas the latter—and our work—identify the fundamental limitations of NS resources in enforcing causal irreversibility within operationally constrained cryptographic protocols.

The causality imposed by NS-boxes and the causality inherent to OT are closely related but fundamentally distinct. NS enforces a negative, structural notion of causality, prohibiting any dependence of one party’s local output distribution on the other party’s input. OT, by contrast, requires a positive, functional notion of causality: the receiver’s private choice must causally determine the reconstructed message. The key observation underlying our impossibility result is that NS correlations cannot be used to enforce such functional causality. Enforcing that a private choice has already been applied would allow the other party to detect this fact and hence violate non-signalling. Consequently, while NS resources may reproduce the statistics of an OT-box, they cannot impose the causal irreversibility required by OT when embedded in a causally ordered communication protocol. This mismatch between structural and functional causality explains both the possibility results based on abstract OT-box simulations and the impossibility results for NS-assisted communication models.

\section{Discussion}\label{Disc}
We investigated the feasibility of network oblivious transfer over NS–assisted noisy channels, focusing on multiple access channels and broadcast channels. Our results demonstrate that network OT is impossible in the presence of tripartite NS correlations, regardless of whether these correlations are trivial or non-trivial. Since both classical and quantum correlations form strict subsets of the NS correlation set, this result implies that any tripartite entanglement shared between the sender(s) and the receiver(s) necessarily leads to information leakage and, consequently, violates the secrecy requirements of OT.

A central question motivating this work is whether quantum correlations (entanglement), or even super-quantum correlations such as PR-boxes, can be leveraged to enhance the OT capacity of shared noisy channels in network settings. Our findings indicate that, within the considered system models, tripartite NS correlations fundamentally undermine OT security and therefore cannot be used to increase OT capacity.

Nevertheless, our analysis suggests that bipartite NS correlations—specifically, correlations shared between the transmitters in the MAC model or between the receivers in the BC model—may still satisfy the OT security constraints when the parties are not allowed to collude. This problem has been considered in this paper for the MAC model. Whether such bipartite NS correlations can be exploited to improve the OT capacity of the underlying noisy channels remains an open problem. Addressing this question constitutes a natural and important direction for future research.
\section*{Acknowledgment}
The authors acknowledge Omar Fawzi and Matthieu Bloch for valuable discussions during the course of this research.

\bibliography{references}

@techreport{Rabin1,
  author       = {Rabin, Michael O.},
  title        = {How to Exchange Secrets by Oblivious Transfer},
  institution  = {Aiken Comput. Lab., Harvard University},
  address      = {Cambridge, MA},
  number       = {TR-81},
  year         = {1981}
}

@article{EGL,
  author       = {Even, Shimon and Goldreich, Oded and Lempel, Abraham},
  title        = {A Randomized Protocol for Signing Contracts},
  journal      = {Communications of the ACM},
  volume       = {28},
  pages        = {637--647},
  year         = {1985}
}

@inproceedings{Dodis,
  author       = {Dodis, Yevgeniy and Ong, S. J. and Prabhakaran, Manoj and Sahai, Amit},
  title        = {On the (im)possibility of cryptography with imperfect randomness},
  booktitle    = {Proceedings of the 45th Annual IEEE Symposium on Foundations of Computer Science (FOCS)},
  pages        = {196--205},
  year         = {2004},
  doi          = {10.1109/FOCS.2004.44}
}

@article{SV,
  author       = {Santha, Miklos and Vazirani, Umesh V.},
  title        = {Generating quasi-random sequences from semi-random sources},
  journal      = {Journal of Computer and System Sciences},
  volume       = {33},
  number       = {1},
  pages        = {75--87},
  year         = {1986}
}

@inproceedings{Joe,
  author       = {Crépeau, Claude and Kilian, Joe},
  title        = {Achieving oblivious transfer using weakened security assumptions},
  booktitle    = {Proceedings of the 29th IEEE Symposium on Foundations of Computer Science (FOCS)},
  address      = {White Plains, NY, USA},
  pages        = {42--52},
  year         = {1988},
  doi          = {10.1109/SFCS.1988.21920}
}

@article{Winter1,
  author       = {Nascimento, Anderson C. A. and Winter, Andreas},
  title        = {On the oblivious transfer capacity of noisy resources},
  journal      = {IEEE Transactions on Information Theory},
  volume       = {54},
  number       = {6},
  pages        = {2572--2581},
  month        = jun,
  year         = {2008}
}

@inproceedings{Winter3,
  author       = {Nascimento, Anderson C. A. and Winter, Andreas},
  title        = {On the oblivious transfer capacity of noisy correlations},
  booktitle    = {Proceedings of IEEE International Symposium on Information Theory (ISIT)},
  address      = {Seattle, WA, USA},
  pages        = {1871--1875},
  year         = {2006}
}

@inproceedings{Rudolf1,
  author       = {Ahlswede, Rudolf and Csiszár, Imre},
  title        = {On oblivious transfer capacity},
  booktitle    = {Proceedings of IEEE Information Theory Workshop on Network Information Theory (ITW)},
  address      = {Volos, Greece},
  pages        = {1--3},
  year         = {2009}
}

@inproceedings{Crepeau3,
  author       = {Crépeau, Claude and Morozov, Kirill and Wolf, Stefan},
  title        = {Efficient unconditional oblivious transfer from almost any noisy channel},
  booktitle    = {Proceedings of SCN},
  volume       = {3352},
  series       = {Lecture Notes in Computer Science},
  address      = {Berlin, Germany},
  publisher    = {Springer-Verlag},
  pages        = {47--59},
  year         = {2005}
}

@inproceedings{Imai,
  author       = {Imai, Hideki and Morozov, Kirill and Nascimento, Anderson C. A.},
  title        = {On the oblivious transfer capacity of the erasure channel},
  booktitle    = {Proceedings of IEEE International Symposium on Information Theory (ISIT)},
  address      = {Seattle, WA, USA},
  pages        = {1428--1431},
  year         = {2006}
}

@inproceedings{Watanabe2,
  author       = {Suda, Shun and Watanabe, Shun and Yamaguchi, Hiroshi},
  title        = {An improved lower bound on oblivious transfer capacity via interactive erasure emulation},
  booktitle    = {Proceedings of IEEE International Symposium on Information Theory (ISIT)},
  pages        = {1872--1877},
  year         = {2024},
  doi          = {10.1109/ISIT57864.2024.10619607}
}

@misc{Watanabe3,
  author       = {Suda, Shun and Watanabe, Shun},
  title        = {An improved lower bound on oblivious transfer capacity using polarization and interaction},
  howpublished = {arXiv:2501.11883 [cs.IT]},
  year         = {2025}
}

@inproceedings{amir1,
  author    = {Shekofteh, A. and Chou, R. A.},
  title     = {{SPIR} with Colluding and Non-Replicated Servers from a Noisy Channel},
  booktitle = {Proceedings of the 60th Annual Allerton Conference on Communication, Control, and Computing},
  year      = {2024},
  pages     = {1--6}
}

@inproceedings{amir2,
  author    = {Shekofteh, A. and Chou, R. A.},
  title     = {Single-Server {SPIR} over Binary Erasure Channels: Benefits of Noisy Side Information},
  booktitle = {Proceedings of the 61st Annual Allerton Conference on Communication, Control, and Computing},
  year      = {2025},
  pages     = {1--5}
}

@inproceedings{amir3,
  author    = {Shekofteh, Amirhossein and Chou, R\'emi A.},
  title     = {Improved Achievable Rate for Single-Server {SPIR} over Binary Erasure Channels},
  booktitle = {Proceedings of the 61st Annual Allerton Conference on Communication, Control, and Computing},
  year      = {2025},
  publisher = {Allerton Conference on Communication, Control, and Computing},
  address   = {Urbana, IL},
  doi       = {}
}

@inproceedings{Hadi2,
  author       = {Aghaee, Hamed and Deppe, Christina},
  title        = {On oblivious transfer capacity of noisy multiple access channel},
  booktitle    = {Proceedings of IEEE International Symposium on Information Theory (ISIT)},
  address      = {Ann Arbor, MI, USA},
  pages        = {1--6},
  year         = {2025},
  doi          = {10.1109/ISIT63088.2025.11195437}
}

@misc{Hadi1,
  author       = {Aghaee, Hamed and Akhbari, Babak and Deppe, Christina},
  title        = {Network oblivious transfer: Limits and capacities},
  howpublished = {arXiv:2501.17021 [cs.IT]},
  year         = {2025}
}

@misc{Hadi3,
  author       = {Aghaee, Hamed and Deppe, Christina and Boche, Holger},
  title        = {Network Oblivious Transfer via Noisy Broadcast Channels},
  howpublished = {arXiv:2510.25343v1 [cs.IT]},
  year         = {2025}
}

@inproceedings{Hadi6,
  author       = {Aghaee, Hamed and Deppe, Christina and Boche, Holger},
  title        = {Oblivious Transfer over Discrete Memoryless Broadcast Channels},
  booktitle= { Proceedings of IEEE International Conference on Communications (ICC2026)},
  year         = {2026, in press}
}

@inproceedings{Hadi4,
  author       = {Aghaee, Hamed and Deppe, Christina and Akhbari, Babak},
  title        = {On Information Theoretical Limits of Oblivious Transfer},
  booktitle    = {Proceedings of IEEE Global Communications Conference (GLOBECOM)},
  address      = {Taipei, Taiwan},
  year         = {2025}
}

@article{Granelli,
  author       = {Granelli, Federico and Bassoli, Roberto and Nötzel, Jan and Fitzek, Frank H. P. and Boche, Holger and da Fonseca, Nuno L. S. and Guerrieri, Alessandro},
  title        = {A Novel Architecture for Future Classical-Quantum Communication Networks},
  journal      = {Wireless Communications \& Mobile Computing},
  volume       = {2022},
  pages        = {Article ID 3770994},
  month        = jan,
  year         = {2022},
  doi          = {10.1155/2022/3770994}
}

@article{Ekert,
  author       = {Ekert, Artur K.},
  title        = {Quantum cryptography based on Bell's theorem},
  journal      = {Physical Review Letters},
  volume       = {67},
  number       = {6},
  pages        = {661--663},
  year         = {1991}
}

@article{Bennett,
  author       = {Bennett, Charles H. and Shor, Peter W. and Smolin, John A. and Thapliyal, Ashish V.},
  title        = {Entanglement-assisted capacity of a quantum channel and the reverse Shannon theorem},
  journal      = {IEEE Transactions on Information Theory},
  volume       = {48},
  number       = {10},
  pages        = {2637--2655},
  year         = {2002}
}

@article{Hsieh,
  author       = {Hsieh, Min-Hsiu and Devetak, Igor and Winter, Andreas},
  title        = {Entanglement-assisted capacity of quantum multiple-access channels},
  journal      = {IEEE Transactions on Information Theory},
  volume       = {54},
  number       = {7},
  pages        = {3078--3090},
  year         = {2008}
}

@ARTICLE{Janis,
  author={Nötzel, Janis},
  journal={IEEE Journal on Selected Areas in Information Theory}, 
  title={Entanglement-Enabled Communication}, 
  year={2020},
  volume={1},
  number={2},
  pages={401-415},
  doi={10.1109/JSAIT.2020.3017121}}

@article{Leditzky,
  author       = {Leditzky, Felix and Alhejji, Mohammed and Levin, Eric and Smith, Graeme},
  title        = {Playing games with multiple access channels},
  journal      = {Nature Communications},
  volume       = {11},
  number       = {1},
  pages        = {1497},
  month        = mar,
  year         = {2020}
}

@article{Seshadri,
  author       = {Seshadri, Ashwin and Leditzky, Felix and Siddhu, Vishal and Smith, Graeme},
  title        = {On the separation of correlation-assisted sum capacities of multiple access channels},
  journal      = {IEEE Transactions on Information Theory},
  volume       = {69},
  number       = {9},
  pages        = {5805--5844},
  month        = sep,
  year         = {2023},
  doi          = {10.1109/TIT.2023.3274434}
}

@article{Pereg1,
  author       = {Pereg, Uriya and Deppe, Christina and Boche, Holger},
  title        = {Quantum interference channels: Capacity bounds and protocols},
  journal      = {IEEE Transactions on Information Theory},
  volume       = {69},
  number       = {5},
  pages        = {1--15},
  year         = {2023}
}

@article{Fawzi,
  author       = {Fawzi, Omar and Fermé, Pierre},
  title        = {Multiple-access channel coding with non-signaling correlations},
  journal      = {IEEE Transactions on Information Theory},
  volume       = {70},
  pages        = {1693--1719},
  year         = {2024}
}

@article{Yard,
  author       = {Yard, Jonathon and Hayden, Patrick and Devetak, Igor},
  title        = {Capacity theorems for quantum multiple-access channels},
  journal      = {IEEE Transactions on Information Theory},
  volume       = {57},
  number       = {10},
  pages        = {7147--7167},
  year         = {2011}
}

@article{Pereg2,
  author       = {Pereg, Uriya and Deppe, Christina and Boche, Holger},
  title        = {Quantum broadcast channels: Capacity and protocols},
  journal      = {IEEE Transactions on Information Theory},
  volume       = {67},
  number       = {12},
  pages        = {1--20},
  year         = {2021}
}

@article{Clauser,
  author       = {Clauser, John F. and Horne, Michael A. and Shimony, Abner and Holt, Richard A.},
  title        = {Proposed experiment to test local hidden-variable theories},
  journal      = {Physical Review Letters},
  volume       = {23},
  number       = {15},
  pages        = {880--884},
  month        = oct,
  year         = {1969},
  doi          = {10.1103/PhysRevLett.23.880}
}

@article{Quek,
  author       = {Quek, Yixian and Shor, Peter W.},
  title        = {Quantum and superquantum enhancements to two-sender, two-receiver channels},
  journal      = {Physical Review A},
  volume       = {95},
  number       = {5},
  pages        = {052329},
  month        = may,
  year         = {2017},
  doi          = {10.1103/PhysRevA.95.052329}
}

@article{Popescu,
  author       = {Popescu, Sandu and Rohrlich, Daniel},
  title        = {Quantum nonlocality as an axiom},
  journal      = {Foundations of Physics},
  volume       = {24},
  number       = {3},
  pages        = {379--385},
  month        = mar,
  year         = {1994},
  doi          = {10.1007/BF02058098}
}

@misc{Hawellek1,
  author       = {Hawellek, Jan and Mohan, Akshay and Aghaee, Hamed and Deppe, Christina},
  title        = {The Interference Channel with Entangled Transmitters},
  howpublished = {arXiv preprint arXiv:2411.10067 [quant-ph]},
  month        = nov,
  year         = {2024},
  doi          = {10.48550/arXiv.2411.10067}
}

@misc{Hawellek2,
  author       = {Hawellek, Jan and Mohan, Akshay and Aghaee, Hamed and Deppe, Christina},
  title        = {The Interference Channel with Entangled Transmitters},
  howpublished = {Accepted for publication in IEEE International Symposium on Information Theory (ISIT) 2025}
}

@INPROCEEDINGS{Hadi5,
  author={Aghaee, Hadi and Deppe, Christian},
  booktitle={2025 IEEE Information Theory Workshop (ITW)}, 
  title={On the Interference Channel with Entangled Transmitters}, 
  year={2025},
  volume={},
  number={},
  pages={1-6},
  doi={10.1109/ITW62417.2025.11240483}}

@inproceedings{Chou,
  author       = {Chou, Robert A.},
  title        = {Pairwise Oblivious Transfer},
  booktitle    = {2020 IEEE Information Theory Workshop (ITW)},
  pages        = {1--5},
  year         = {2021}
}

@article{Wolf,
  author  = {Stefan Wolf and J{\"u}rg Wullschleger},
  title   = {Oblivious Transfer and Non-Locality},
  journal = {IEEE Transactions on Information Theory},
  volume  = {52},
  number  = {6},
  pages   = {2519--2528},
  year    = {2006},
  doi     = {10.1109/TIT.2006.874980}
}

@article{Short,
  author  = {Anthony J. Short and Nicolas Gisin and Sandu Popescu},
  title   = {The physics of no-bit-commitment: generalised quantum non-locality versus oblivious transfer},
  journal = {Quantum Information Processing},
  volume  = {5},
  number  = {2},
  pages   = {131--143},
  year    = {2006},
  doi     = {10.1007/s11128-006-0015-4}
}

@article{Buhrman,
    author = {Buhrman, Harry and Christandl, Matthias and Unger, Falk and Wehner, Stephanie and Winter, Andreas},
    title = {Implications of superstrong non-locality for cryptography},
    journal = {Proceedings of the Royal Society A: Mathematical, Physical and Engineering Sciences},
    volume = {462},
    number = {2071},
    pages = {1919-1932},
    year = {2006},
    month = {02},
    issn = {1364-5021},
    doi = {10.1098/rspa.2006.1663},
    url = {https://doi.org/10.1098/rspa.2006.1663},
    eprint = {https://royalsocietypublishing.org/rspa/article-pdf/462/2071/1919/674011/rspa.2006.1663.pdf},
}

\end{document}